\documentclass[conference,a4paper]{IEEEtran}
\IEEEoverridecommandlockouts
\usepackage{cite}
\usepackage{amsmath,amssymb,amsfonts}
\usepackage{algorithmic}
\usepackage{amsthm}
\allowdisplaybreaks
\usepackage{graphicx}
\usepackage{subfigure} 
\usepackage{textcomp}
\usepackage{xcolor}
\usepackage{bbm}
\newtheorem{theorem}{\bf Theorem}
\newtheorem{lemma}{Lemma}

\newtheorem{Remark}{Remark}

\newtheorem{Definition}{Definition}
\newtheorem*{Threat Model}{Threat Model}
\newtheorem{corollary}{\bf Corollary}
\usepackage{amsmath,amsfonts}
\usepackage{algorithmic}
\usepackage{array}
\usepackage[caption=false,font=normalsize,labelfont=sf,textfont=sf]{subfig}
\usepackage{stfloats}
\usepackage{url}
\usepackage{verbatim}
\usepackage{graphicx}
\usepackage{balance}
\def\BibTeX{{\rm B\kern-.05em{\sc i\kern-.025em b}\kern-.08em
		T\kern-.1667em\lower.7ex\hbox{E}\kern-.125emX}}
\begin{document}

\title{Multi-Server Secure Aggregation with Unreliable Communication Links
}
\author{Kai Liang, Songze Li, Ming Ding and Youlong Wu
	\thanks{Kai Liang and Youlong Wu are with the School of Information Science and Technology, ShanghaiTech University, Shanghai 201210, China. (e-mail: \{liangkai, wuyl1\}@shanghaitech.edu.cn). Songze Li is with the Thrust of Internet of Things, The Hong Kong University of Science and Technology (Guangzhou), Guangzhou, China, and also with the Department of Computer Science and Engineering, The Hong Kong University of Science and Technology, Hong Kong SAR, China (e-mail: songzeli@ust.hk).
		Ming Ding is with the Data61, CSIRO, Sydney, NSW 2015, Australia (e-mail: ming.ding@data61.csiro.au).}}
\maketitle

\maketitle

\begin{abstract}

In many distributed learning setups such as federated learning (FL),  client nodes at the edge use individually collected data to compute local gradients and  send them to a central master server. The  master server then aggregates the received gradients and broadcasts the aggregation to all clients, with which the clients can update the global model. In this paper, we consider multi-server federated learning with secure aggregation and unreliable communication links. We first define a threat model using Shannon’s information-theoretic security framework, and propose a novel scheme called Lagrange Coding with Mask (LCM), which divides the servers into groups and uses Coding and Masking techniques. LCM can achieve a trade-off between the uplink and downlink communication loads by adjusting the number of servers in each group. Furthermore, we 
derive the lower bounds of the uplink and downlink communication loads, respectively, and prove that   LCM achieves the optimal uplink communication load, which is unrelated to the number of collusion clients.
\end{abstract}

\begin{IEEEkeywords}
Coding computing,  federated learning, straggling links, secure aggregation
\end{IEEEkeywords}

\section{Introduction}

With the proliferation of smartphones, wearables, and other Internet of Things (IoT) devices, an enormous amount of data is generated every moment. 
Many applications require this vast amount of data to improve their services. 
Machine learning (ML) techniques can make predictions or inferences from massive amounts of data, making it widely used in many applications. However, the ML paradigm creates serious privacy and security issues due to the large amount of data used about users, so careful security mechanisms are required to ensure that private information leakage is minimized while maintaining algorithm performance as much as possible. Due to the increasing storage and computing power of mobile devices, a new distributed learning paradigm has attracted much attention, i.e., Federated Learning (FL)\cite{11}. In FL, client nodes at the edge compute the local gradients and send them to a central master server, and the master aggregates them and broadcasts the aggregation to all clients, with which the clients can update the global model.  An important feature of federated learning is that data is stored locally, and only gradient values are exchanged during the joint training, which protects users' privacy to a certain extent.

In classic FL, a single server runs the model or gradient aggregation, which can lead to a single point of failure. To address this issue, multiple servers can be used for aggregation operations to improve the robustness of federated learning\cite{he2020secure, xu2020privacy,brunetta2021non,cryptoeprint:2022/1695,corrigan2017prio,cryptoeprint:2021/576}. The multi-server secure aggregation problem still suffers from communication bottlenecks due to limited communication resources and the unreliability of communication links. In \cite{1}, motivated by the emergence of a multi-access edge computing ecosystem, Prakash \emph{et al.}  used helper nodes and coding strategies to achieve resiliency against \emph{straggling} client-to-helpers links.  Note that their model requires the master to communicate with helpers through noiseless links, and the two-hop communications from clients to helpers and helpers to master may result in a larger communication delay. The authors of \cite{sasidharan2022coded} established a trade-off between the communication costs at clients and helper nodes by using the well-known pyramid codes. In \cite{1,sasidharan2022coded}, they considered the case where communication links between clients and servers (helper nodes) are unreliable, and used coding techniques to inject redundancy to resist straggling communication links. Still, the authors did not consider privacy constraints. On the other hand, the authors of \cite{he2020secure, xu2020privacy,cryptoeprint:2022/1695,corrigan2017prio,cryptoeprint:2021/576} considered privacy-preserving machine learning with two servers. \cite{jia2022x,9834439,9965815} consider the problem of federated submodel learning under multiple servers and the authors proposed an adaptive scheme which utilized all available servers to reduce their communication costs. The authors\cite{brunetta2021non} proposed a non-interactive, secure verifiable aggregation for decentralized, privacy-preserving learning (NIVA), which allowed distributed aggregation of secret inputs from multiple clients by multiple untrusted servers. However, these works generally assume that the communication links between clients and servers were reliable and error-free.

In this paper, we consider the \emph{multi-server secure aggregation problem} for FL, that is, clients send gradient values to multiple servers, and then each client can get the final aggregation result without revealing its private data. Our goal is to design a communication efficient and robust secure aggregation scheme when the communication links between servers and clients are unreliable. 

The main contributions of this work are summarized:
\begin{itemize}
\item We formally characterize the multi-server secure aggregation problem with unreliable communication links, and define a threat model using Shannon’s information-theoretic security framework\cite{shannon1949communication}. To the best of our knowledge, we are the first to consider unreliable communication links in a multi-server secure aggregation scenario. In the threat model, the colluding servers cannot infer any information about client's data, even the aggregated value, and meanwhile, the colluding clients cannot know other honest clients' data.
\item We propose a novel scheme called Lagrange Coding with Mask (LCM). The LCM is motivated by Masking techniques and the Lagrange Coded Computing (LCC)\cite{7}, which was initially developed for distributed computing applications and has demonstrated remarkable resilience against straggling and malicious nodes. However, it is worth noting that LCC is not directly applicable to the scenarios addressed in this paper since colluding clients may potentially gain access to the original information of other clients through the servers. LCM can trade off the uplink communication load for the downlink communication load by adjusting the number of servers in each group.
\item We also characterize the lower bounds of the optimal uplink and downlink communication loads under the threat model. 
When there is only one server in each group, LCM achieves the optimal uplink communication load,   at the cost of a higher downlink communication load.
As the number of collusion servers increases, the optimal uplink communication load also increases. Interestingly, the optimal uplink communication load is unrelated to the number of collusion clients.


\end{itemize}

\section{Problem Setup}

{
\subsection{Federated Learning with Multiple Servers}
 Consider a  multi-server federated learning setup consisting of straggling communication links as depicted in Fig. \ref{model}.  It consists of $E$ clients and $H$ servers, where each client has a unicast communication link to each server, while the links between clients and servers are unreliable in the sense that
there exists up to  $s$ straggling links out of $H$ server links per client. 
\begin{figure}[htb]
	\centerline
	{\includegraphics[width=1\linewidth]{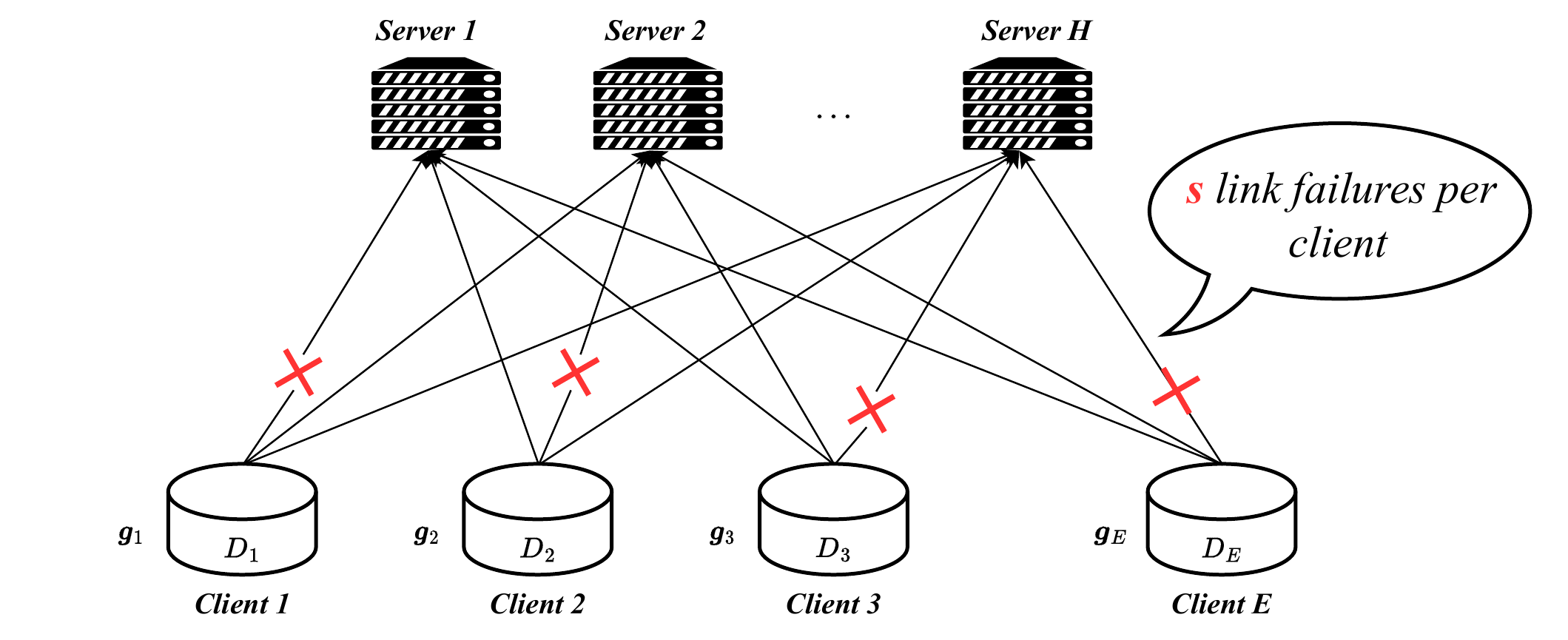}}
	\caption{Multi-server federated learning network with straggling links.}
	\label{model}
\end{figure}

Use notation $[E]\triangleq \{1,\ldots,E\}$ and denote the data set as   $D=\{(\boldsymbol{x}_{j},y_{j}):\boldsymbol{x}_{j}\in\mathbb{R}^{d},y_{j}\in\mathbb{R},j=1,2,\cdots, n \}$, where $d$ is a positive integer,  $\boldsymbol{x}_{j}$ an $y_{j}$ denote $j$th data and its label, respectively.  Each client $i\in[E]$ collects local private data ${D_{i}}$ with $D=\bigcup_{i=1}^{E}{D_{i}}$. The training process is to solve the following optimization problem:
\begin{equation}\label{LossFun}
	\theta^{*}=\arg\min_{\theta\in\mathbb{R}^{p}}\sum_{i=1}^{E}\sum_{(\boldsymbol{x}_{j},y_{j})\in D_i}{\ell_i(\theta;\boldsymbol{x}_{j},y_{j})+\lambda R(\theta)},
\end{equation}
where $\ell_i(\cdot)$ denotes the loss function of Client $i$, $\theta$ denotes the global model parameters,  $p$ is the length of the training model, $R(\cdot)$ is the regularization function, and $\lambda$  is  the regularization parameter.

We apply the widely used  gradient descent (GD) algorithm  to solve \eqref{LossFun}.  Let $\boldsymbol{g}^{(t)}_{i}\in\mathbb{R}^{p}$ be the local gradient associated with $D_{i}$   at iteration $t$, i.e., $\boldsymbol{g}^{(t)}_{i}=\sum_{(\boldsymbol{x}_{j},y_{j})\in D_{i}}\nabla\ell_i(\theta^{(t)};\boldsymbol{x}_{j},y_{j})$. The corresponding global gradient at iteration $t$ is
\begin{equation}\label{GlobalG}
\boldsymbol{g}^{(t)}_{D}=\sum_{i=1}^{E}{\boldsymbol{g}^{(t)}_{i}}=\sum_{i=1}^{E}\sum_{(\boldsymbol{x}_{j},y_{j})\in D_i}(\nabla\ell_i(\theta^{(t)};\boldsymbol{x}_{j},y_{j})+\lambda\nabla R(\theta^{(t)}),
\end{equation}
where $\theta^{(t)}$ denotes the global model parameters at iteration $t$ and is updated in the following way: 
\begin{equation}
	\theta^{(t+1)}=U(\theta^{(t)},\boldsymbol{g}^{(t)}_{D}),
\end{equation}
where $U$ denotes the  gradient-based optimizer. To simplify  notations, we omit the superscript $(t)$ in $\boldsymbol{g}^{(t)}_i$ and $\boldsymbol{g}_{D}^{(t)}$. We also assume that each element of gradients $\boldsymbol{g}_i$ and $\boldsymbol{g}_{D}$ is a symbol from a finite field $\mathbb{F}$ of size $|\mathbb{F}|$. We also assume that each gradient $\boldsymbol{g}_{i}$ is a uniform distribution over the field $\mathbb{F}^p$ and $\{\boldsymbol{g}_i\}_{i\in[E]}$ are independent. The uniformity and independence of the gradients are required for the converse proof, but are not necessary for the achievability proof\cite{zhao2022information}.

\subsection{Network Model}
The communication between the clients and servers can be described as a failure table ${\bf{T}}=({\bf{T}}_{ij})_{i\in[E],j\in[H]}$ with $E$ rows (clients) and $H$ columns (servers), where ${\bf{T}}_{ij}=1$ means that Client $i$ and Server $j$ are successfully connected, while ${\bf{T}}_{ij}=0$ is the opposite. Assume that there are at most $s$ straggler links per client, which means that each row of table $T$ has at most $s$ zeros. Note that the failure table ${\bf{T}}$ is unknown to the servers and clients until they communicate.   {The servers report back to successfully connected clients which clients they received messages from.} Then, each Client $i\in[E]$ knows ${\bf{T}}_{:,j_1},\cdots,{\bf{T}}_{:,j_k}$, where ${\bf{T}}_{:,j_l}$ denotes the column $j_l$ of failure table ${\bf{T}}$ for $l\in[k]$ and $j_1,\cdots,j_k$ denote the indices of the servers successfully connected to Client $i$, i.e., ${\bf{T}}_{ij_{l}}=1$ for $l\in[k], k\geq H-s$.

Each Client $i$  encodes $\boldsymbol{g}_{i}$ using encoding functions $\mathcal{F}_i=(\mathcal{F}_{i}^{1},\mathcal{F}_{i}^{2},\cdots, \mathcal{F}_{i}^{H})$ and a random bit sequence $\boldsymbol{r}_i\in \{0,1\}^{*}$ independent of gradient $\boldsymbol{g}_{i}$  to generate the coded message $\boldsymbol{c}_{i}=[\boldsymbol{c}_{i,1}^{T},\boldsymbol{c}_{i,2}^{T},\cdots,\boldsymbol{c}_{i,H}^{T}]$ with $\boldsymbol{c}_{i}\in\mathbb{F}^{q_i}$ and $\boldsymbol{c}_{i,j}\in\mathbb{F}^{q^j_i}$, 
where $\{0,1\}^{*}$ denotes all binary sequences, $\mathcal{F}_{i}^{j}$ is a mapping function from $\mathbb{F}^{p}\times \{0,1\}^{*}$ to $\mathbb{F}^{q^j_{i}}$, $\boldsymbol{c}_{i,j}=\mathcal{F}_{i}^{j}(\boldsymbol{g}_{i},\boldsymbol{r}_i)$,  and $q_i=\sum_{j=1}^{H}{q_{i}^j}$, for $j\in[H]$. Then Client $i$ sends $\boldsymbol{c}_{i,j}$ to Server $j$. { We can use $q_i^j$ to denote the length of messages $\boldsymbol{c}_{i,j}$.}

Let $\Omega(s)$ be the set of all straggling patterns with  $s$ straggling links per client. We assume equal probability for any pattern $m\in\Omega(s)$, i.e., $P(m)=\frac{1}{|\Omega(s)|}$. Given a straggling pattern $m$, let $\mathcal{L}_i^{j,m}(\boldsymbol{c}_{i,j})$ be $0$ if the link between Client $i$ and Server $j$ is straggling, and be $\boldsymbol{c}_{i,j}$ otherwise. After the uplink transmission, Server $j$ receives  the concatenating  messages:
\begin{IEEEeqnarray}{rCl}\label{eqHerlpM}
	{\bf{u}}^{j,m}=  (\mathcal{L}_1^{j,m}(\boldsymbol{c}_{1,j}),\ldots,\mathcal{L}_E^{j,m}(\boldsymbol{c}_{E,j})).
\end{IEEEeqnarray}
Each Server $j$ encodes ${\bf{u}}^{j,m}$ using a downlink encoder $\mathcal{H}^{j,m}_i$ to generate message $W^{j,m}_i =\mathcal{H}^{j,m}_i({\bf{u}}^{j,m})$ with $W^{j,m}_i\in\mathbb{F}^{d^{j,m}_i}$ for Client $i$. Then, 
Client $i$ receives messages  
\begin{IEEEeqnarray}{rCl}\label{eqUserM}
	W_{i}^m =(\mathcal{L}_i^{1,m}(W^{1,m}_i  ),\ldots,\mathcal{L}_i^{H,m}(W^{H,m}_i )),
\end{IEEEeqnarray} 
with $W_{i}^m\in \mathbb{F}^{ d_i^m}$ and  { $d_i^m\triangleq \sum_{j=1}^H d^{j,m}_i\mathbbm{1}_{\{\mathcal{L}_i^{j,m}(\boldsymbol{c}_{i,j})\neq 0\}}$} where we use $d^{j,m}$ to denote the length of messages $W_{i}^m$. Finally Client $i$  uses decoder  $\mathcal{D}_i$ to recover $\boldsymbol{g}_{D}=\mathcal{D}_i(W_{i}^m)$, i.e.,
\begin{equation}\label{decode}
	H(\boldsymbol{g}_{D}|W_{i}^m)=0, \   \forall i\in[E].
\end{equation}


\subsection{Privacy Leakage}
{Since the local gradients carry a lot of information about the clients' dataset, it is easy to reconstruct the clients' data through the local gradients by using a model inversion attack\cite{zhu2019deep,geiping2020inverting}. In order to avoid this privacy leakage problem, The authors\cite{bonawitz2017practical} introduced a protocol for secure aggregation with one server which can compute the aggregated gradient while ensuring that the server (and other clients) can only learn the value of the aggregated gradient, and no additional information about the client’s data can be learned. Similarly, we can define the secure aggregation protocol for the multi-server federated learning setup. Formally, we consider the following threat model:
	\begin{Threat Model}\label{tm1}
		We assume that all nodes (clients and servers) are semi-honest, i.e., all nodes will run the program according to the protocol, but each server is curious about all clients' private data and each client is curious about other clients' data. We assume that there are at most $T_h$ servers colluding and at most $T_c$ clients colluding. These colluding servers can neither infer any client's data nor infer the final aggregated value,  and  meanwhile, colluding clients cannot know other honest clients' data.
	\end{Threat Model}
	\begin{Definition}\label{D1} $(T_h, T_c)$-private: We say that an aggregation protocol is $(T_h, T_c) $-private if it satisfies the following conditions:\\
		(1): The mutual information between the messages owned by any set of at most $T_h$ colluding servers and the data of clients is zero, i.e.,
		\begin{equation}\label{pr1}
		I(\boldsymbol{g}_1,\cdots,\boldsymbol{g}_E;\mathtt{VIEW}_{\mathcal{T}_h})=0,
		\end{equation}
		where $\mathtt{VIEW}_{\mathcal{T}_h}$ denotes all messages known to the any subset $\mathcal{T}_h$ of servers with $|\mathcal{T}_h|\leq T_h$, i.e., $\mathtt{VIEW}_{\mathcal{T}_h}\triangleq\{{\bf{u}}^{j,m}:j\in\mathcal{T}_h\}$, for  straggling pattern $m\in\Omega(s)$;\\
		(2): Given any at most $T_c$ colluding clients, the mutual information between the messages owned by the colluding clients and the data of  honest clients is zero, conditioning on gradients of colluding clients and the aggregation gradient, i.e.,
		\begin{equation}\label{pri22}
		I(\{\boldsymbol{g}_i\}_{i\in[E]\backslash \mathcal{T}_c};\mathtt{VIEW}_{\mathcal{T}_c}|(\boldsymbol{g}_j)_{j\in\mathcal{T}_c},\boldsymbol{g}_D)=0,
		\end{equation}
		where $\mathtt{VIEW}_{\mathcal{T}_c}$ denotes all messages known to the any subset $\mathcal{T}_c$ of clients with $|\mathcal{T}_c|\leq T_c$, i.e.,  $\mathtt{VIEW}_{\mathcal{T}_c}\triangleq\{{\bf{u}}^{j,m}:j\  \text{satisfies}\  T_{ij}=1\ \text{for some}\  i\in\mathcal{T}_c \}\bigcup\{(\boldsymbol{g}_i,\boldsymbol{r}_i):i\in\mathcal{T}_c\}$, for  straggling pattern $m\in\Omega(s)$,
	\end{Definition}
}
\begin{Remark}
	{Our setup allows clients to retrieve any data stored on the servers with which they have established a successful connection and we can regard the client as colluding with the servers it successfully connects to. Therefore, colluding clients can potentially access all data residing on the servers.}
\end{Remark}

\subsection{Problem Formulation}
Our goal is to design aggregation and communication schemes such that all clients can recover the global gradient $\boldsymbol{g}_{D}$  under all scenarios of up to $s$ straggling communication links per client, and the threat model in Definition \ref{D1}. }

Now we define some metrics related to our problem similar to {\cite{1,sasidharan2022coded}}:

\begin{Definition}Resiliency threshold $s$: If each client can recover $\boldsymbol{g}_{D}$ in \eqref{GlobalG} for any pattern of up to $s$ straggling links per client, we say the scheme has a resiliency threshold of $s\in\mathbb{N}^{+}$ or is $s$-resiliency.
\end{Definition}

\begin{Definition}
Uplink communication load $C_\textnormal{up}$: Given a uplink transmission with   encoded message $(\boldsymbol{c}_1,\ldots,\boldsymbol{c}_E)$, $\boldsymbol{c}_i\in \mathbb{F}^{q_{i}} $, $C_\textnormal{up}$ is the maximum size of these coded messages $\{\boldsymbol{c}_i\}_{i=1}^E$, normalized by the length of the local gradient vector, i.e.,
\begin{equation}\label{comupdef}
C_\textnormal{up}=\max_{i\in[E]} \frac{q_i}{p}.
\end{equation}
\end{Definition}

\begin{Definition} Downlink communication load $C_\textnormal{down}$: Given a straggling pattern $m\in\Omega(s)$ and downlink encoded messages  $W^m_i\in  \mathbb{F}^{ d_i^m}$, for $i\in[E]$, $C_\textnormal{down}$ is the maximum size of $W_{i}^m$, normalized by the lengeth of the local gradient vector, i.e.,
\begin{equation}\label{rCl}
C_\textnormal{down}=\max_{i\in[E]} \frac{d_i^m}{p}.
\end{equation}

\end{Definition}

{

For a resiliency threshold of $s$ and straggling pattern $m$, a pair $(C_\textnormal{up},C_\textnormal{down})  $ is \emph{achievable} if there exists a \emph{$(T_h, T_c)  $-private} aggregation protocol with encoders $\{\mathcal{F}_i\}_{i=1}^E$, $\{\mathcal{H}^{j,m}_i\}_{i\in[E],j\in[H]}$  and decoders $\{\mathcal{D}_{i}\}_{i=1}^E$  that achieves communication loads $C_\textnormal{up}$ and $C_\textnormal{down}$.

For the federated learning setting with $E$ clients, $H$ servers, a resiliency threshold of $s$, and privacy thresholds $T_h$ and $T_c$, the problem is to characterize the minimum uplink and downlink communication loads as follows:
\begin{eqnarray}
C_\textnormal{up}^{*}  &=&\inf_{(C_\textnormal{up},C_\textnormal{down})  \in\mathcal{A}}\sup_{m\in\Omega(s)} C_\textnormal{up},\label{st11}\\
C_\textnormal{down}^{*}  &=&\sup_{m\in\Omega(s)}\inf_{(C_\textnormal{up},C_\textnormal{down})  \in\mathcal{A}} C_\textnormal{down},\label{st21}
\end{eqnarray}
where $\mathcal{A}$ is the set of all achievable $(C_\textnormal{up},C_\textnormal{down})$ over all $m\in\Omega(s)$.} 

\section{Main Results}
The following {Lemma} \ref{uppri} gives upper bounds on the number of servers and clients that can collude under the threat model, respectively.
\begin{lemma}\label{uppri}
	For a federated learning setting with $E$ clients, $H$ servers, {the length  of gradient $p$} and a resiliency threshold of $s<\frac{H}{2}$, if a secure aggregation protocol is $(T_h,T_c)  $-private and $s$-resiliency, then we have
\begin{equation}\label{cup2}
T_h\leq H-2s-1,\  \ T_c\leq E-2.
\end{equation}

\end{lemma}
\begin{proof}
	See Appendix \ref{l1proof}
\end{proof}

The following Theorem \ref{the2} characterizes the optimal uplink and downlink communication loads under the threat model.
\begin{theorem}\label{the2}
	Given $T_h$ and $T_c$, for a federated learning setting with $E$ clients, $H$ servers, {length  of gradient $p$} and a resiliency threshold of $s<\frac{H}{2}$, we have
	\begin{align}
	\frac{H}{H-2s-T_h}\leq &C_\textnormal{up}^{*}\leq C_\textnormal{up}^\textnormal{LCM},\label{eqUp2}\\
	\frac{H-2s}{H-2s-T_h}\leq &C_\textnormal{down}^{*}\leq C_\textnormal{down}^\textnormal{LCM},\label{eqDown2}
	\end{align}
	where $C_\textnormal{up}^\textnormal{LCM}=\frac{\lfloor\frac{H}{v}\rfloor v}{\lfloor\frac{H}{v}\rfloor-\lfloor\frac{2s}{v}\rfloor-T_h}, C_\textnormal{down}^\textnormal{LCM}=\frac{\lfloor \frac{H}{v}\rfloor-\lfloor \frac{2s}{v}\rfloor}{\lfloor\frac{H}{v}\rfloor-\lfloor\frac{2s}{v}\rfloor-T_h}(E-1-\lceil\frac{E-1}{{{\lfloor \frac{H}{v}\rfloor-\max\{s-\lfloor\frac{H}{v}\rfloor (a-1),0 \}} \choose {\lfloor \frac{H}{v}\rfloor-\lfloor \frac{2s}{v}\rfloor}}}\rceil+v)$
	and the parameter ${v}$ satisfies $\lfloor\frac{H}{v}\rfloor-\lfloor\frac{2s}{v}\rfloor-T_h\geq 1$.
	
\end{theorem}

\begin{proof}
	The uplink and downlink communication loads in \eqref{eqUp2} and \eqref{eqDown2} are achieved based on the LCM scheme present in Section \ref{LCM}. The lower bound is given in Appendix \ref{SecApp} and the proof of security is given in Appendix \ref{secproof}. 
\end{proof}
{\begin{Remark} 
	Our approach is motivated by the Lagrange Coded Computing (LCC) paradigm \cite{7}, which was initially developed for distributed computing applications and has demonstrated remarkable resilience against straggling and malicious nodes. However, it is worth noting that LCC is not directly applicable to the scenarios addressed in this paper. Specifically, in LCC, the original data can be reconstructed via Lagrangian interpolation once the data from sufficient working nodes is collected. If LCC is utilized directly, colluding clients may potentially gain access to the original information of other clients through the servers.
\end{Remark}
\begin{Remark} 
	The parameter ${v}$ is a hyperparameter of the LCM scheme present in Section \ref{LCM}. Specifically, the LCM scheme divides the servers into multiple groups, and each group contains ${v}$ servers. Therefore, the number of groups can be adjusted via the parameter $a$, thereby trading off between the LCM scheme's uplink and downlink communication loads.
	
\end{Remark}
From Theorem \ref{the2}, we obtain the folowing Corollary.
\begin{corollary} \label{optimality}
	If $v=1$, then $C_\textnormal{up}^{*}=\frac{H}{H-2s-T_h}$, i.e., LCM scheme can achieve the optimal uplink communication load, and $C_\textnormal{down}^{*}$ is upper bounded by 
	\begin{equation}\label{ourdo1}
	C_\textnormal{down}^{*}\leq \frac{H-2s}{H-2s-T_h}(E-\lceil\frac{E-1}{{{H-s}\choose {s}}}\rceil).
	\end{equation}
	If $v=2s+1$,  then $C_\textnormal{up}^{*}$ is upper bounded by 
	\begin{equation}\label{ourUp2}
	C_\textnormal{up}^{*}\leq \frac{\lfloor\frac{H}{2s+1}\rfloor}{\lfloor\frac{H}{2s+1}\rfloor-T_h}(2s+1),
	\end{equation} and $C_\textnormal{down}^{*}$ is upper bounded by 
	\begin{equation}\label{ourdo2}
	C_\textnormal{down}^{*}\leq\frac{\lfloor\frac{H}{2s+1}\rfloor}{\lfloor\frac{H}{2s+1}\rfloor-T_h}(2s+1).
	\end{equation}
\end{corollary}
	It can be seen that when the number of clients $E$ is large enough, the downlink communication load of the LCM scheme with $v=1$ is higher than that of the LCM scheme with $v=2s+1$, however, the LCM scheme with $v=1$ can achieve the optimal uplink communication load. As the collusion number $T_h$ increases, the lower bounds of the optimal uplink and downlink communication loads also increase. Interestingly, the number of collusion clients $T_c$ does not affect both the lower and upper bounds of the communication loads.
\begin{Remark} 
	Note that the parameter ${v}$ satisfies $\lfloor\frac{H}{v}\rfloor-\lfloor\frac{2s}{v}\rfloor-T_h\geq 1$. In other words, if the parameter ${v}$ is fixed, for the LCM scheme, the number of colluding servers $T_h$ should satisfy $T_h\leq \lfloor\frac{H}{v}\rfloor-\lfloor\frac{2s}{v}\rfloor-1$. Particularly, if $v=1$, we have $T_h\leq H-2s-1$, this means LCM scheme present in Section \ref{LCM} can be resistant to the most colluding servers.
	
\end{Remark}}

\section{Secure Aggregation Scheme}
\subsection{Lagrange Coding with Mask (LCM)}\label{LCM}
First, we divide the $H$ servers into $\lfloor \frac{H}{v}\rfloor$ groups $\mathcal{G}_1,\cdots,\mathcal{G}_{\lfloor \frac{H}{v}\rfloor}$ of $v$ servers each. If $v$ does not divide $H$, we remove $(H-\lfloor \frac{H}{v}\rfloor a)$ servers.\\
\textbf{Encoding at the clients:} Each pair of clients $(i,j), i<j$ should agree on some random vector $\boldsymbol{s}_{ij}\in \mathbb{F}^p$. Then, Client $i\in[E]$ computes:
\begin{equation}\label{mask12}
\boldsymbol{y}_i=\boldsymbol{g}_i+\sum_{j\in [E], i<j}\boldsymbol{s}_{ij}-\sum_{j\in [E], i>j}\boldsymbol{s}_{ji} \ \ (\textnormal{mod}\
|\mathbb{F}|).
\end{equation}

Client $i$ partitions its masking gradient update into $k$ components as $\boldsymbol{y}_{i}=[\boldsymbol{y}_{i,1}^{T},\cdots ,\boldsymbol{y}_{i,k}^{T}]^{T}$ where  $\boldsymbol{y}_{i,r}^{T}\in\mathbb{F}^{\frac{p}{k}}$ for $r\in[k]$ and $k=\lfloor \frac{H}{v}\rfloor-\lfloor \frac{2s}{v}\rfloor-T_h$. Then, Client $i$ randomly chooses $T_h$ vectors $\boldsymbol{Z}_{i,1},\cdots, \boldsymbol{Z}_{i,T_h}$ from $\mathbb{F}^{\frac{p}{k}}$. Similar to the LCC, Client $i$ needs to construct a polynomial $u_i: \mathbb{F}\rightarrow\mathbb{F}^{\frac{p}{k}}$. We can select any $k+T_h$ distinct elements $\beta_1,\cdots,\beta_{k+T_h}$ from $\mathbb{F}$ such that $u_i(\beta_r)=\boldsymbol{y}_{i,r}$, for any $r\in[k]$ and $u_i(\beta_r)=\boldsymbol{Z}_{i,r-k}$ for any $r\in\{k+1,\cdots,k+T_h\}$. By using the Lagrange interpolation polynomial, Client $i$ can compute
\begin{equation}
\begin{aligned}
u_i(x)\triangleq &\sum_{r=1}^{k} \boldsymbol{y}_{i,r}\cdot\prod_{l\in[k+T_h]\backslash \{r\}}\frac{x-\beta_l}{\beta_r-\beta_l}+\\
&\sum_{r=k+1}^{k+T_h} \boldsymbol{Z}_{i,r-k}\cdot\prod_{l\in[k+T_h]\backslash \{r\}}\frac{x-\beta_l}{\beta_r-\beta_l}.
\end{aligned}
\end{equation}
We then select $\lfloor \frac{H}{v}\rfloor$ distinct elements $\alpha_j,j\in[\lfloor \frac{H}{v}\rfloor]$ from $\mathbb{F}$ where we assume that $\{\alpha_j\}_{j\in[\lfloor \frac{H}{v}\rfloor]}\cap\{\beta_r\}_{r\in[k+T_h]}=\emptyset$. Then, Client $i$ can encode its gradient as $\{\hat{\boldsymbol{g}}_{i,j}\triangleq u_i(\alpha_j)\}_{j\in[\lfloor \frac{H}{v}\rfloor]}$. Finally, Clietnt $i$ sends $\hat{\boldsymbol{g}}_{i,j}$ to all servers in the Group $\mathcal{G}_j$.

We note that the encoded encrypted data sent to servers is linear combinations of $\{\boldsymbol{y}_{i,r}\}_{r\in[k]}$ and $\{\boldsymbol{Z}_{i,r}\}_{r\in[T_h]}$, i.e.,
\begin{equation*}
\begin{pmatrix}
u_i(\alpha_1)\\
u_i(\alpha_2)\\
\vdots\\
u_i(\alpha_{\lfloor \frac{H}{v}\rfloor})\\
\end{pmatrix}=
\begin{pmatrix}
U_{11}&\cdots&U_{1(k+T_h)}\\
U_{21}&\cdots&U_{2(k+T_h)} \\
\vdots&\ddots&\vdots\\
U_{\lfloor \frac{H}{v}\rfloor1}&\cdots&U_{\lfloor \frac{H}{v}\rfloor(k+T_h)}\\
\end{pmatrix}
\begin{pmatrix}
\boldsymbol{y}_{i,1}\\
\vdots\\
\boldsymbol{y}_{i,k}\\
\boldsymbol{Z}_{i,1}\\
\vdots\\
\boldsymbol{Z}_{i,T_h}\\
\end{pmatrix},
\end{equation*}
where $U_{jr}\triangleq\prod_{l\in[k+T_h]\backslash \{r\}}\frac{\alpha_j-\beta_l}{\beta_r-\beta_l}.$ By the construction, we have the following result for the uplink communication load:
\begin{equation}
C_\textnormal{up}=\lfloor \frac{H}{v}\rfloor\frac{v}{k}=\frac{\lfloor \frac{H}{v}\rfloor v}{\lfloor \frac{H}{v}\rfloor-\lfloor \frac{2s}{v}\rfloor-T_h}.
\end{equation}
{\begin{Remark}
	The proposed method  necessitates secure exchange of $p$-dimensional random vectors $\boldsymbol{s}_{ij}$ between each pair of clients $(i,j)$. An established approach to overcome this challenge is information-theoretic secure key agreement through accessing correlated random variables \cite{256484}. In practice, communication overhead can be minimized by having each pair of clients agree on a common seed for a pseudo-random generator (PRG). As a result, Client $i$ and $j$ can generate the mask vector $\boldsymbol{s}_{ij}$ based on their shared seed, thus negating the need to directly agree on $\boldsymbol{s}_{ij}$. A random seed for each pair of clients $(i,j)$ can be agreed upon using a key exchange protocol, such as the Diffie-Hellman key agreement \cite{DHKEY76}. Notably, the PRG-generated random vectors are computationally indistinguishable from authentic uniformly distributed random vectors. Hence, the proposed scheme in this paper can be demonstrated to provide computational security, supported by corresponding performance guarantees.
\end{Remark}}
\textbf{Aggregation:} According to the properties of Lagrange Coded Computing, the coefficients of $u_i(x)$ can be calculated from any $k+T_h$ messages in $\{u_i(\alpha_j)\}_{j\in[\lfloor \frac{H}{v}\rfloor]}$. Given any  $ m\in\Omega(s)$, for each Client $i$, there are at least $\max \{s-\lfloor\frac{H}{v}\rfloor (v-1),0\}$ groups that cannot connect to Client $j$, that is, all servers in these groups cannot connect to  Client $i$. On the other hand, among the remaining $\lfloor \frac{H}{v}\rfloor-\max\{s-\lfloor\frac{H}{v}\rfloor (v-1),0\}$ groups, we can always find $\lfloor \frac{H}{v}\rfloor-\lfloor \frac{2s}{v}\rfloor$ groups that can connect Client $i$ and Client $j$ at the same time, that is, there are some servers in these $\lfloor \frac{H}{v}\rfloor-\lfloor \frac{2s}{v}\rfloor$  groups that receive messages from Client $i$ and Client $j$ at the same time. Then, each client find the maximum number $M$ of other clients whose messages are received by the same $\lfloor \frac{H}{v}\rfloor-\lfloor \frac{2s}{v}\rfloor$ groups. We use $\mathcal{M}$ to denote these $M$ clients. Then, servers in these groups aggregate coded messages from $\mathcal{M}$ and send them to Client $i$. Thus, Client $i$ can recover $\sum_{i\in \mathcal{M}}{\boldsymbol{y}_{i}}$. For the remaining $(E-M-1)$ clients' gradient, the corresponding $\lfloor \frac{H}{v}\rfloor-\lfloor \frac{2s}{v}\rfloor$ servers simply forward their messages to the Client $i$. We can bound $M$ by using the balls and bins problem. If we have $E-1$ balls denoting $E-1$ clients and any ${\lfloor \frac{H}{v}\rfloor-\lfloor \frac{2s}{v}\rfloor}$ groups can be seen as a bin. There are ${{\lfloor \frac{H}{v}\rfloor-\max\{s-\lfloor\frac{H}{v}\rfloor (v-1),0 \}} \choose {\lfloor \frac{H}{v}\rfloor-\lfloor \frac{2s}{v}\rfloor}}$ such bins. So, there must be a bin with at least $\frac{E-1}{{{\lfloor \frac{H}{v}\rfloor-\max\{s-\lfloor\frac{H}{v}\rfloor (v-1),0 \}} \choose {\lfloor \frac{H}{v}\rfloor-\lfloor \frac{2s}{v}\rfloor}}}$ balls, i.e., $M\geq \lceil\frac{E-1}{{{\lfloor \frac{H}{v}\rfloor-\max\{s-\lfloor\frac{H}{v}\rfloor (v-1),0 \}} \choose {\lfloor \frac{H}{v}\rfloor-\lfloor \frac{2s}{v}\rfloor}}}\rceil$. Finally, we obtain the following bound:
\begin{equation*}
C_\textnormal{down}\leq \frac{\lfloor \frac{H}{v}\rfloor-\lfloor \frac{2s}{v}\rfloor}{k}(E-1-\lceil\frac{E-1}{{{\lfloor \frac{H}{v}\rfloor-\max\{s-\lfloor\frac{H}{v}\rfloor (v-1),0 \}} \choose {\lfloor \frac{H}{v}\rfloor-\lfloor \frac{2s}{v}\rfloor}}}\rceil+v).
\end{equation*}	

\textbf{Example 1.} We consider an example with $v=1$, $E=4$, $H=6$, $T_h=2$, $s=1$, and $T_c=2$. 
\begin{figure}[htbp]
	\centering
	\subfigure[Messages table]{
		\includegraphics[width=0.466\linewidth]{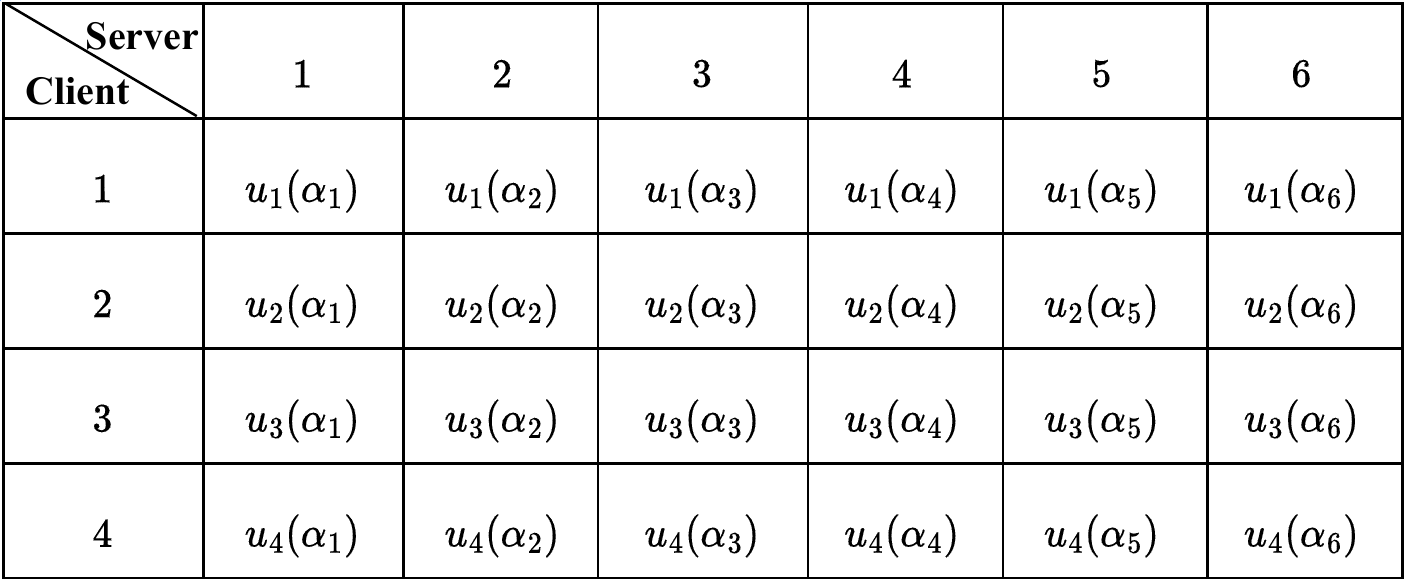}
	}
	\subfigure[Failure scenario 1]{
		\includegraphics[width=0.466\linewidth]{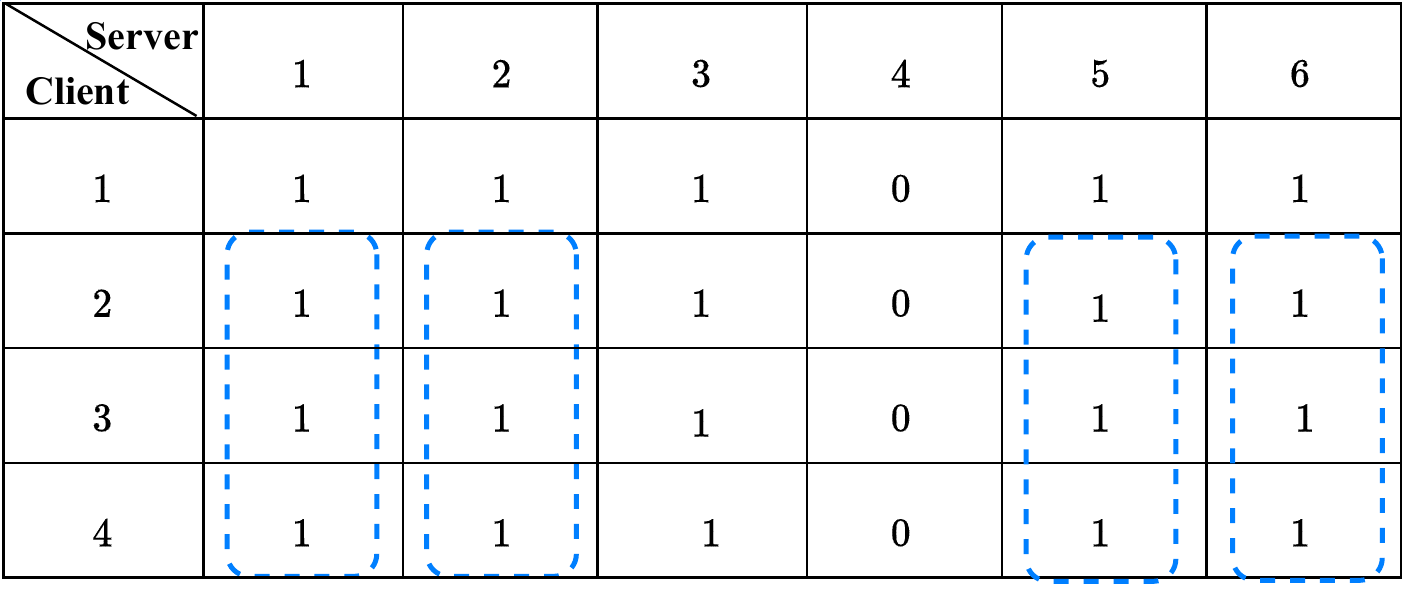}
	}
	\quad
	\subfigure[Failure scenario 2]{
		\includegraphics[width=0.466\linewidth]{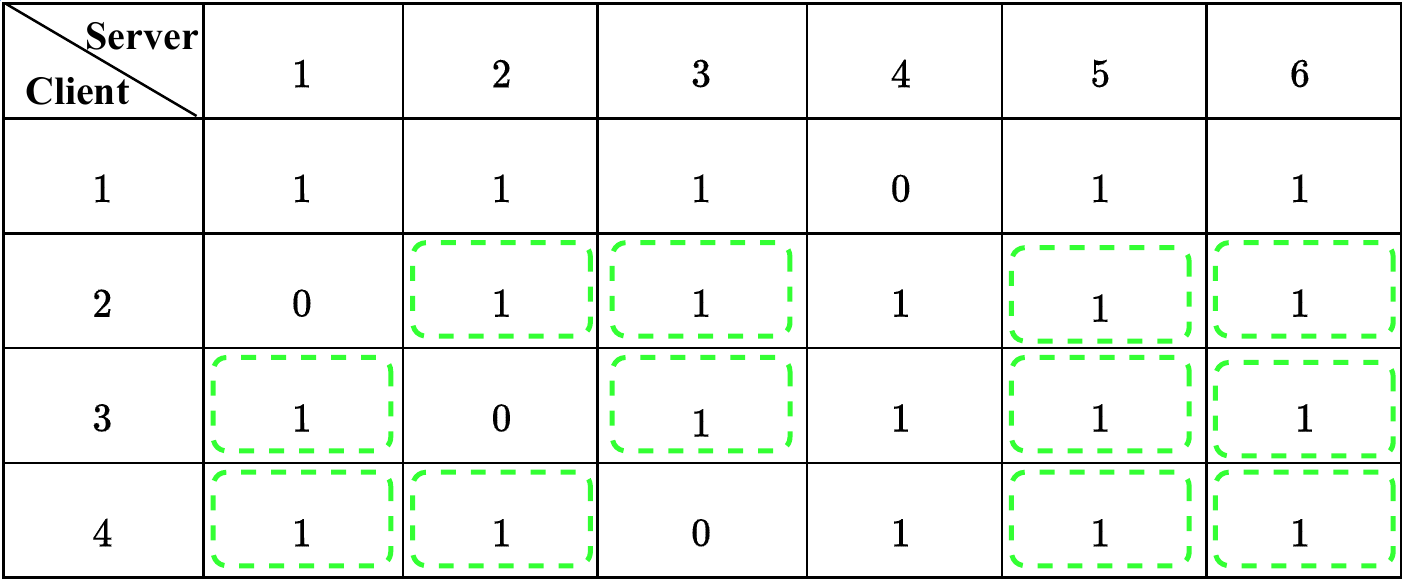}
	}
	\subfigure[Failure scenario 3]{
		\includegraphics[width=0.466\linewidth]{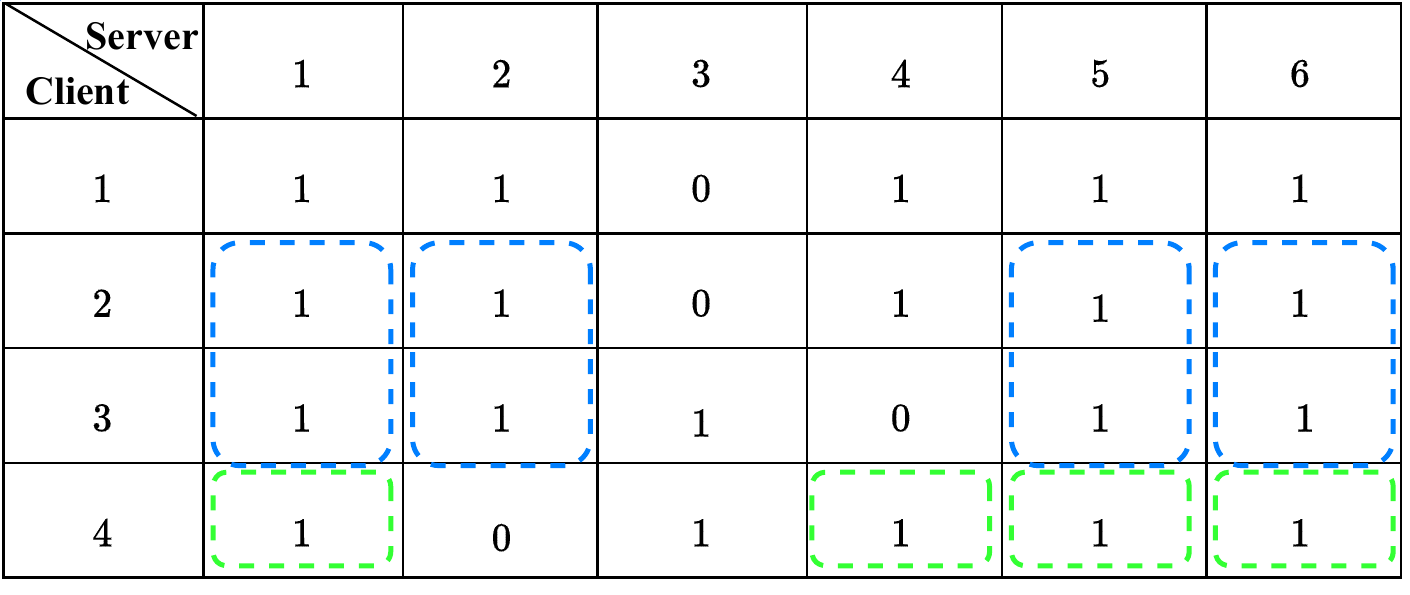}
	}
	\caption{Messages table and Failure table}
	\label{fig}
\end{figure}
Each Client $i\in[4]$ computes $\boldsymbol{y}_i$ as (\ref{mask12}) and generates two random vectors $\boldsymbol{Z}_{i,1}$ and $\boldsymbol{Z}_{i,2}$. Then, each Client $i\in[4]$ partitions its masking gradient updateinto $2$ components as $\boldsymbol{y}_{i}=[\boldsymbol{y}_{i,1}^{T},\boldsymbol{y}_{i,2}^{T}]^{T}$. Let $\beta_1=1,\beta_2=2, \beta_3=3,\beta_4=4$ and $\alpha_1=5,\alpha_2=6, \alpha_3=7,\alpha_4=8,\alpha_5=9,\alpha_6=10$, each Client $i\in[4]$ can compute
\begin{equation*}
\begin{pmatrix}
u_i(\alpha_1)\\
u_i(\alpha_2)\\
u_i(\alpha_3)\\
u_i(\alpha_4)\\
u_i(\alpha_5)\\
u_i(\alpha_6)\\
\end{pmatrix}=
\begin{pmatrix}
-1&4&-6&4\\
-4&15&-20&10\\
-10&36&-45&20\\
-20&70&-84&35\\
-35&120&-140&56\\
-56&-189&-216&84\\
\end{pmatrix}
\begin{pmatrix}
\boldsymbol{y}_{i,1}\\
\boldsymbol{y}_{i,2}\\
\boldsymbol{Z}_{i,1}\\
\boldsymbol{Z}_{i,2}\\
\end{pmatrix}.
\end{equation*}
It is easy to calculate the uplink communication load of each client as $C_\textnormal{up}=3$. Then, clients find the maximum number $M$ of other clients whose messages are received by the same $4$ servers. The same decoding criterion is applied to partially aggregated components, i.e., for any partial aggregation at a group, partial aggregations over the same clients (rows in the failure table) at other groups must be received by clients. The entries within the rows (highlighted in blue) are partially aggregated at the respective servers and sent to clients. The remaining entries (highlighted in green) are simply forwarded to clients.

In Fig. 2(b), it is sufficient for Servers $1,2,5,6$ to aggregate their received messages to obtain $\boldsymbol{g}_{D,1}$ and $\boldsymbol{g}_{D,2}$. In Fig. 2(c), the servers cannot do partial aggregation, so all entries are simply sent to Client $1$. In Fig. 2(d), Server $1,2,5$ and $6$ aggregate their received messages to obtain $u_2(5)+u_3(5)$, $u_2(6)+u_3(6)$, $u_2(9)+u_3(9)$ and $u_2(10)+u_3(10)$ respectively and send them to Client $1$. For the gradient of Client $4$, Server $1,4,5$ and $6$ send $u_4(5)$, $u_4(8)$, $u_4(9)$ and $u_4(10)$ to Client $1$, respectively.
For Client $1$, the downlink communication loads in the three scenarios are $2,6$ and $4$, respectively.

\textbf{Example 2.} We consider an example with $v=3$, $E=4$, $H=6$, $T_h=1$, $s=1$ and $T_c=2$. 
\begin{figure}[htbp]
	\centering
	\subfigure[Messages table]{
		\includegraphics[width=0.466\linewidth]{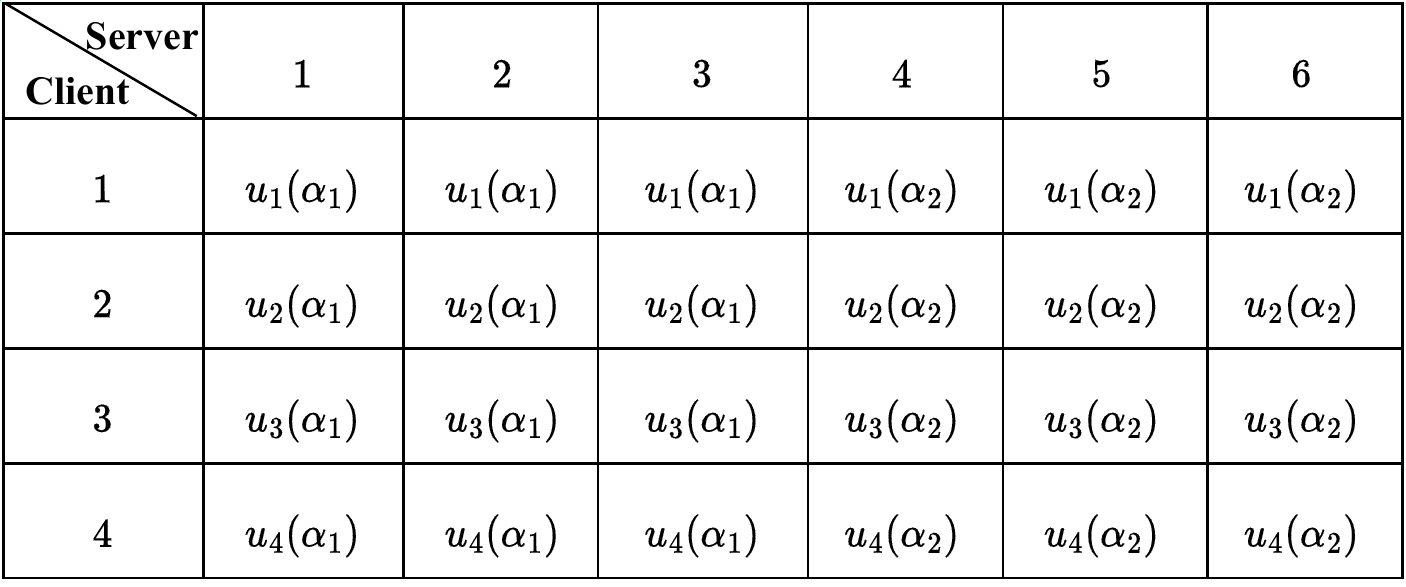}
	}
	\subfigure[Failure scenario 1]{
		\includegraphics[width=0.466\linewidth]{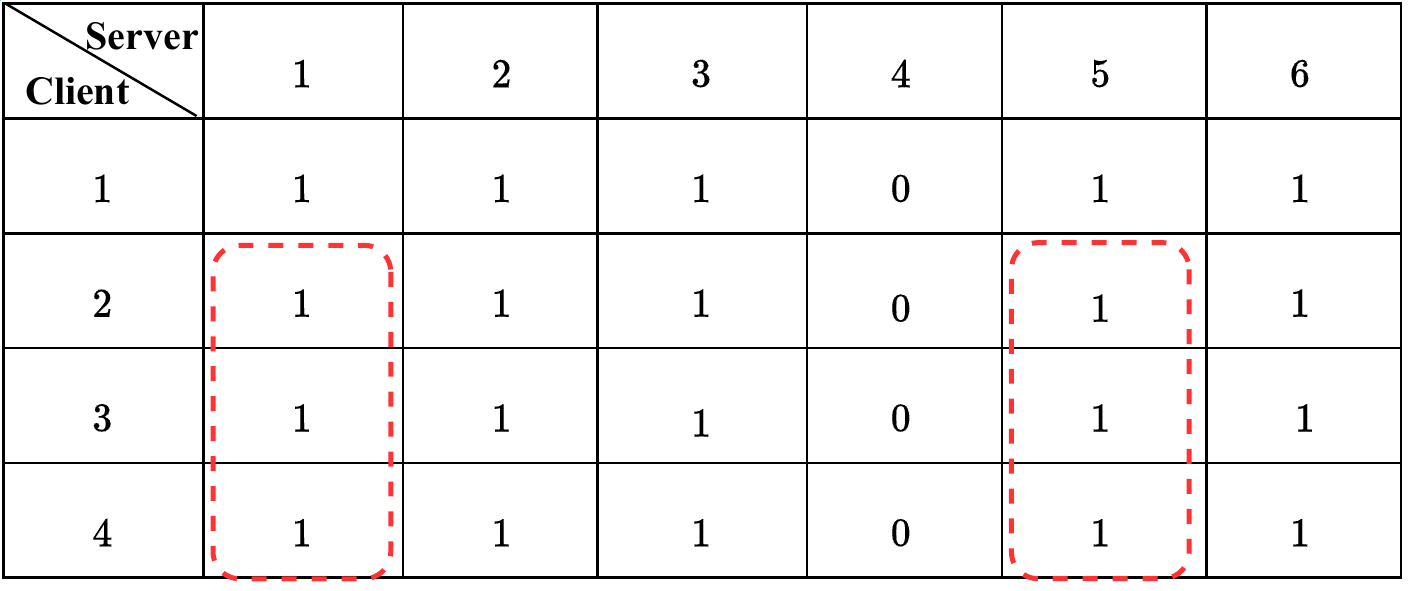}
	}
	\quad
	\subfigure[Failure scenario 2]{
		\includegraphics[width=0.466\linewidth]{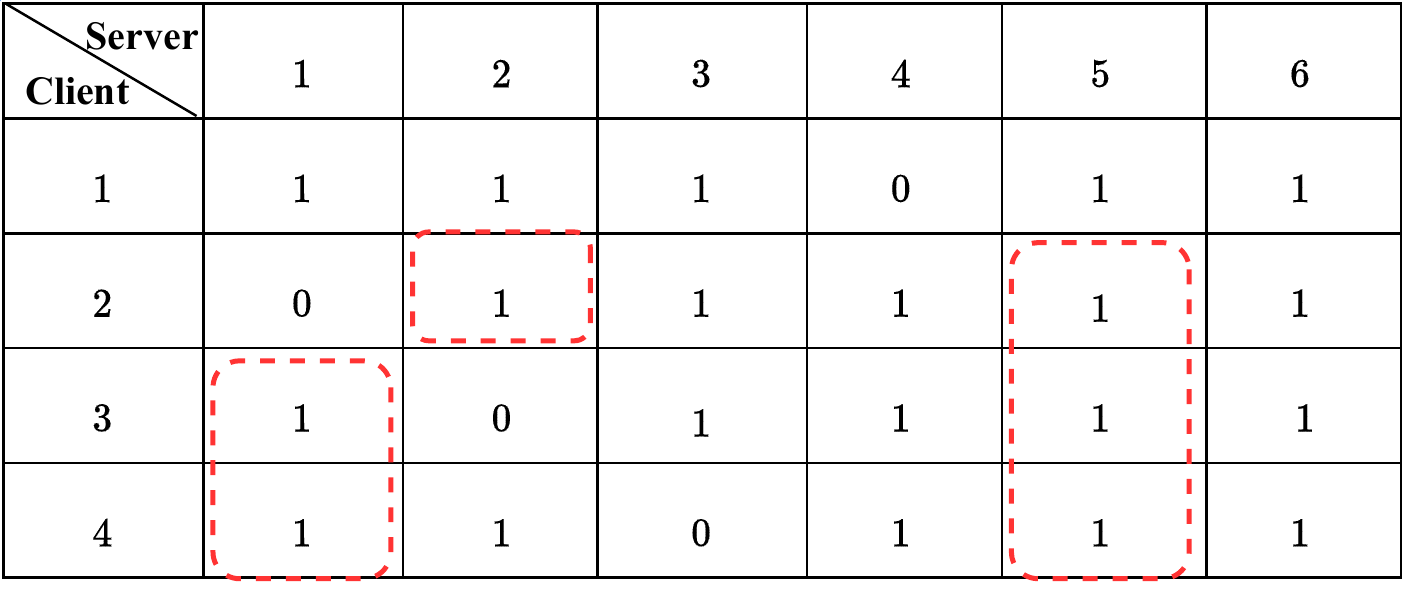}
	}
	\subfigure[Failure scenario 3]{
		\includegraphics[width=0.466\linewidth]{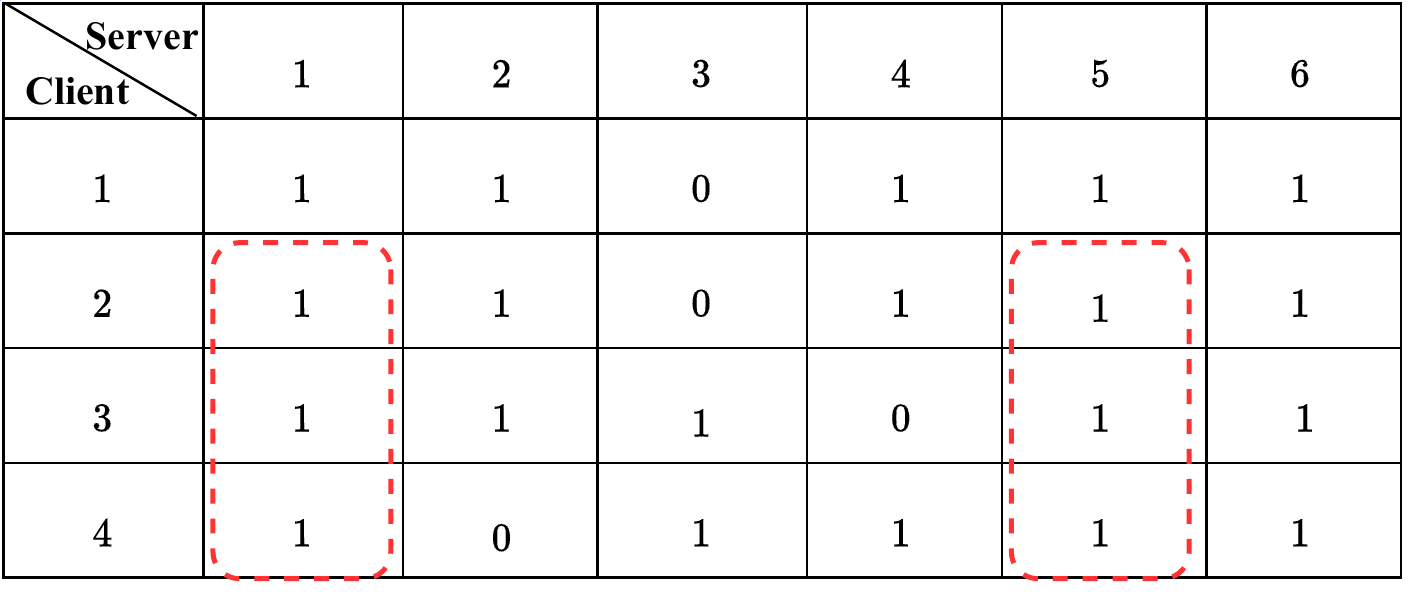}
	}
	\caption{Messages table and Failure table}
	\label{fig}
\end{figure}
Each Client $i\in[4]$ computes $\boldsymbol{y}_i$ as (\ref{mask12}) and generates a random vectors $\boldsymbol{Z}_{i}$. Let $\beta_1=1,\beta_2=2$ and $\alpha_1=3,\alpha_2=4$, each Client $i\in[4]$ can compute $u_i(3)=-\boldsymbol{y}_i+2\boldsymbol{Z}_i$ and $u_i(3)=-2\boldsymbol{y}_i+3\boldsymbol{Z}_i$.
It is easy to calculate the uplink communication load of each client as $C_\textnormal{up}=6$.

In Fig. 3(b) and 3(d), Server $1$ and Server $5$ aggregate their received messages to obtain $u_2(3)+u_3(3)+u_4(3)$ and $u_2(4)+u_3(4)+u_4(4)$, respectively and send them to Client $1$. In Fig. 3(c), Server $1$ and Server $2$ send $u_3(3)+u_4(3)$ and $u_2(3)$ to Client $1$, respectively and Server $5$ sends $u_2(4)+u_3(4)+u_4(4)$. For Client $1$, the downlink communication loads in the three scenarios are $2,3$ and $2$, respectively.
\section{Conclusion}
In this paper, we consider the multi-server secure aggregation problem with unreliable communication links. For the investigated threat model defined by using Shannon’s information-theoretic security framework, we characterize the lower bounds of uplink and downlink communication loads. We also proposed a novel scheme called Lagrange Coding with Mask (LCM) which achieves a trade-off between the uplink and downlink communication loads by adjusting the number of servers in each group. In particular, when there was only one server in each group, LCM achieved the optimal uplink communication load, but at the cost of a higher downlink communication load.
\bibliographystyle{IEEEtran.bst}
\bibliography{1.bib}

\begin{thebibliography}{10}
\providecommand{\url}[1]{#1}
\csname url@samestyle\endcsname
\providecommand{\newblock}{\relax}
\providecommand{\bibinfo}[2]{#2}
\providecommand{\BIBentrySTDinterwordspacing}{\spaceskip=0pt\relax}
\providecommand{\BIBentryALTinterwordstretchfactor}{4}
\providecommand{\BIBentryALTinterwordspacing}{\spaceskip=\fontdimen2\font plus
\BIBentryALTinterwordstretchfactor\fontdimen3\font minus
  \fontdimen4\font\relax}
\providecommand{\BIBforeignlanguage}[2]{{%
\expandafter\ifx\csname l@#1\endcsname\relax
\typeout{** WARNING: IEEEtran.bst: No hyphenation pattern has been}%
\typeout{** loaded for the language `#1'. Using the pattern for}%
\typeout{** the default language instead.}%
\else
\language=\csname l@#1\endcsname
\fi
#2}}
\providecommand{\BIBdecl}{\relax}
\BIBdecl

\bibitem{11}
\BIBentryALTinterwordspacing
B.~McMahan, E.~Moore, D.~Ramage, S.~Hampson, and B.~A. y~Arcas,
  ``{Communication-efficient learning of deep networks from decentralized
  data},'' ser. Proceedings of Machine Learning Research, A.~Singh and J.~Zhu,
  Eds., vol.~54.\hskip 1em plus 0.5em minus 0.4em\relax Fort Lauderdale, FL,
  USA: PMLR, 20--22 Apr 2017, pp. 1273--1282. [Online]. Available:
  \url{http://proceedings.mlr.press/v54/mcmahan17a.html}
\BIBentrySTDinterwordspacing

\bibitem{he2020secure}
L.~He, S.~P. Karimireddy, and M.~Jaggi, ``Secure byzantine-robust machine
  learning,'' \emph{arXiv preprint arXiv:2006.04747}, 2020.

\bibitem{xu2020privacy}
G.~Xu, H.~Li, Y.~Zhang, S.~Xu, J.~Ning, and R.~H. Deng, ``Privacy-preserving
  federated deep learning with irregular users,'' \emph{IEEE Transactions on
  Dependable and Secure Computing}, vol.~19, no.~2, pp. 1364--1381, 2022.

\bibitem{brunetta2021non}
C.~Brunetta, G.~Tsaloli, B.~Liang, G.~Banegas, and A.~Mitrokotsa,
  ``Non-interactive, secure verifiable aggregation for decentralized,
  privacy-preserving learning,'' in \emph{Australasian Conference on
  Information Security and Privacy}.\hskip 1em plus 0.5em minus 0.4em\relax
  Springer, 2021, pp. 510--528.

\bibitem{cryptoeprint:2022/1695}
\BIBentryALTinterwordspacing
M.~Rathee, C.~Shen, S.~Wagh, and R.~A. Popa, ``Elsa: Secure aggregation for
  federated learning with malicious actors,'' Cryptology ePrint Archive, Paper
  2022/1695, 2022, \url{https://eprint.iacr.org/2022/1695}. [Online].
  Available: \url{https://eprint.iacr.org/2022/1695}
\BIBentrySTDinterwordspacing

\bibitem{corrigan2017prio}
H.~Corrigan-Gibbs and D.~Boneh, ``Prio: Private, robust, and scalable
  computation of aggregate statistics.'' in \emph{NSDI}, 2017, pp. 259--282.

\bibitem{cryptoeprint:2021/576}
\BIBentryALTinterwordspacing
S.~Addanki, K.~Garbe, E.~Jaffe, R.~Ostrovsky, and A.~Polychroniadou, ``Prio+:
  Privacy preserving aggregate statistics via boolean shares,'' Cryptology
  ePrint Archive, Paper 2021/576, 2021, \url{https://eprint.iacr.org/2021/576}.
  [Online]. Available: \url{https://eprint.iacr.org/2021/576}
\BIBentrySTDinterwordspacing

\bibitem{1}
S.~{Prakash}, A.~{Reisizadeh}, R.~{Pedarsani}, and A.~S. {Avestimehr},
  ``Hierarchical coded gradient aggregation for learning at the edge,'' in
  \emph{2020 IEEE International Symposium on Information Theory (ISIT)}, June
  2020, pp. 2616--2621.

\bibitem{sasidharan2022coded}
B.~Sasidharan and A.~Thomas, ``Coded gradient aggregation: A tradeoff between
  communication costs at edge nodes and at helper nodes,'' \emph{IEEE Journal
  on Selected Areas in Communications}, vol.~40, no.~3, pp. 761--772, 2022.

\bibitem{jia2022x}
Z.~Jia and S.~A. Jafar, ``X-secure t-private federated submodel learning with
  elastic dropout resilience,'' \emph{IEEE Transactions on Information Theory},
  vol.~68, no.~8, pp. 5418--5439, 2022.

\bibitem{9834439}
S.~Vithana and S.~Ulukus, ``Private read update write (pruw) with storage
  constrained databases,'' in \emph{2022 IEEE International Symposium on
  Information Theory (ISIT)}, 2022, pp. 2391--2396.

\bibitem{9965815}
------, ``Private federated submodel learning with sparsification,'' in
  \emph{2022 IEEE Information Theory Workshop (ITW)}, 2022, pp. 410--415.

\bibitem{shannon1949communication}
C.~E. Shannon, ``Communication theory of secrecy systems,'' \emph{The Bell
  system technical journal}, vol.~28, no.~4, pp. 656--715, 1949.

\bibitem{7}
\BIBentryALTinterwordspacing
Q.~Yu, S.~Li, N.~Raviv, S.~M.~M. Kalan, M.~Soltanolkotabi, and S.~A.
  Avestimehr, ``Lagrange coded computing: Optimal design for resiliency,
  security, and privacy,'' ser. Proceedings of Machine Learning Research,
  K.~Chaudhuri and M.~Sugiyama, Eds., vol.~89.\hskip 1em plus 0.5em minus
  0.4em\relax PMLR, 16--18 Apr 2019, pp. 1215--1225. [Online]. Available:
  \url{http://proceedings.mlr.press/v89/yu19b.html}
\BIBentrySTDinterwordspacing

\bibitem{zhao2022information}
Y.~Zhao and H.~Sun, ``Information theoretic secure aggregation with user
  dropouts,'' \emph{IEEE Transactions on Information Theory}, vol.~68, no.~11,
  pp. 7471--7484, 2022.

\bibitem{zhu2019deep}
L.~Zhu, Z.~Liu, and S.~Han, ``Deep leakage from gradients,'' \emph{Advances in
  neural information processing systems}, vol.~32, 2019.

\bibitem{geiping2020inverting}
J.~Geiping, H.~Bauermeister, H.~Dr{\"o}ge, and M.~Moeller, ``Inverting
  gradients-how easy is it to break privacy in federated learning?''
  \emph{Advances in Neural Information Processing Systems}, vol.~33, pp.
  16\,937--16\,947, 2020.

\bibitem{bonawitz2017practical}
K.~Bonawitz, V.~Ivanov, B.~Kreuter, A.~Marcedone, H.~B. McMahan, S.~Patel,
  D.~Ramage, A.~Segal, and K.~Seth, ``Practical secure aggregation for
  privacy-preserving machine learning,'' in \emph{proceedings of the 2017 ACM
  SIGSAC Conference on Computer and Communications Security}, 2017, pp.
  1175--1191.

\bibitem{kai2023}
K.~Liang, S.~Li, M.~Ding, and Y.~Wu, ``Communication-efficient multi-server
  secure aggregation with unreliable communication links,'' \emph{arXiv
  preprint arXiv:xxxx.xxxx}, 2023.

\bibitem{256484}
U.~Maurer, ``Secret key agreement by public discussion from common
  information,'' \emph{IEEE Transactions on Information Theory}, vol.~39,
  no.~3, pp. 733--742, 1993.

\bibitem{DHKEY76}
W.~Diffie and M.~Hellman, ``New directions in cryptography,'' \emph{IEEE
  Transactions on Information Theory}, vol.~22, no.~6, pp. 644--654, 1976.

\end{thebibliography}
\newpage
\appendix
\subsection{Converse Proof}\label{SecApp}

We first present a useful lemma for the converse proof.
\begin{lemma}\label{Lema1}
	\emph{Let $X^{n}$ be $n$-dimensional discrete random vector,  $J\subset[n]$ be a random variable independent of $X^{n}$ with $|J|=k$. Denote $X_J\triangleq(X_j:j\in J)$,  and denote 
		$H(X_{J}|J)$ the average entropy of a $k$-subset of coordinates.  
		Then \begin{equation}\label{lem}
		H(X_{J}|J)\leq\max_{i\in[n]}\mathbb{P}[j\in J]\sum\nolimits_{j=1}^{n}{H(X_{j})}.
		\end{equation}}
\end{lemma}  
\begin{proof}
	\begin{IEEEeqnarray*}{rCl}
		H(X_{J}|J)&\overset{(a)}{=}&\sum_{i=1}^{n\choose k}{\mathbb{P}(J=J_{i})H(X_{J_i}|J=J_{i})}\\
		&{=}&\sum_{i=1}^{n\choose k}{\mathbb{P}(J=J_{i})H(X_{i_{1}},\cdots,X_{i_{|J|}}|J=J_{i})}\\
		&\overset{(b)}{=}&\sum_{i=1}^{n\choose k}{\mathbb{P}(J=J_{i})}\sum_{j=1}^{|J|}{H(X_{i_{j}}|X_{i_{1}},\cdots,X_{i_{j\!-\!1}},J\!=\!J_{i})}\\
		&\overset{(c)}{\leq}&\sum_{i=1}^{n\choose k}{\mathbb{P}(J=J_{i})}\sum_{j=1}^{|J|}{H(X_{i_{j}}|J=J_{i})}\\
		&\overset{(d)}{=}&\sum_{u=1}^{n}{(\sum_{u\ \in\ J_{i}}{\mathbb{P}(J=J_{i}))}H(X_{u})}\\
		&{=}&\sum_{u=1}^{n}{\mathbb{P}(u\in J)H(X_{u})}\\
		&\leq&\max_{u\in[n]}{\mathbb{P}(u\in J)}\sum_{u=1}^{n}{H(X_{u})},
	\end{IEEEeqnarray*}
	where $(a)$ holds by the definition of conditional entropy; $(b)$ holds by the chain rule of entropy; $(c)$ holds by conditioning reduces entropy; in $(d)$, we rearrange each items ${\mathbb{P}(J=J_{i}))}H(X_{J_i}|J=J_{i})$ and use the fact that $J\subset[n]$ is independent of $X^{n}$.
\end{proof}

{We assume that each gradient $\boldsymbol{g}_{i}$ is a uniform distribution over the field $\mathbb{F}^p$ and $\{\boldsymbol{g}_i\}_{i\in[E]}$ are independent. The uniformity and independence of the gradients are required for the converse proof, but are not necessary for the achievability proof. To simplify the notation, we omit the superscript $m$.}

Next, we give the lower bounds of $C_\textnormal{up}^{*}$ and $C_\textnormal{down}^{*}$. First, we assume that $j_1,\ldots,j_{H-s}\in[H]$ denote the indices of the server that successfully send messages to Client $i$. Consider any $H-2s-T_h$ servers in $j_1,\ldots,j_{H-s}$ denoted by $\{j_{l_k}\}_{k\in[H-2s-T_h]}$ and the rest of the servers are denoted by $\{j_{l_{H-2s-T_h+1}},\cdots j_{l_{H-s}}\}$, we have 
\begin{align}\label{endeq2}
H&\big(W_i^{j_{l_1}},\cdots, W_i^{j_{l_{H-2s-T_h}}}|(\boldsymbol{g}_i, \boldsymbol{r}_i:i\in[E]\backslash\{t\})\big)\notag\\
&\ {\geq} H\big(W_i^{j_{l_1}},\cdots, W_i^{j_{l_{H-2s-T_h}}}|(\boldsymbol{g}_i, \boldsymbol{r}_i:i\in[E]\backslash\{t\}),\notag\\
&\qquad \{W_i^{j_{l}}\}_{l=l_{H-2s-T_h+1},\cdots,l_{H-s}}\big)\notag\\
&\ {=}I\big(W_i^{j_{l_1}},\cdots, W_i^{j_{l_{H-2s-T_h}}}; \boldsymbol{g}_t|(\boldsymbol{g}_i, \boldsymbol{r}_i:i\in[E]\backslash\{t\}),\notag\\
&\qquad \{W_i^{j_{l}}\}_{l=l_{H-2s-T_h+1},\cdots,l_{H-s}}\big)\notag\\
&\  +H\big(W_i^{j_{l_1}},\cdots, W_i^{j_{l_{H-2s-T_h}}}|\boldsymbol{g}_t, (\boldsymbol{g}_i, \boldsymbol{r}_i:i\in[E]\backslash\{t\}),\notag\\
&\qquad \{W_i^{j_{l}}\}_{l=l_{H-2s-T_h+1},\cdots,l_{H-s}}\big)\notag\\
&\ {\geq} I\big(W_i^{j_{l_1}},\cdots, W_i^{j_{l_{H-2s-T_h}}}; \boldsymbol{g}_t|(\boldsymbol{g}_i, \boldsymbol{r}_i:i\in[E]\backslash\{t\}),\notag\\
&\qquad \{W_i^{j_{l}}\}_{l=l_{H-2s-T_h+1},\cdots,l_{H-s}}\big)\notag\\
&\  +H\big(W_i^{j_{l_1}},\cdots, W_i^{j_{l_{H-2s-T_h}}}|(\boldsymbol{g}_i, \boldsymbol{r}_i:i\in[E]),\notag\\
&\qquad \{W_i^{j_{l}}\}_{l=l_{H-2s-T_h+1},\cdots,l_{H-s}}\big)\notag\\
&\overset{(a)}{=}I\big(W_i^{j_{l_1}},\cdots, W_i^{j_{l_{H-2s-T_h}}}; \boldsymbol{g}_t|(\boldsymbol{g}_i, \boldsymbol{r}_i:i\in[E]\backslash\{t\}),\notag\\
&\qquad \{W_i^{j_{l}}\}_{l=l_{H-2s-T_h+1},\cdots,l_{H-s}}\big)\notag\\
&\ {=}H\big(\boldsymbol{g}_t|(\boldsymbol{g}_i, \boldsymbol{r}_i:i\in[E]\backslash\{t\}), \{W_i^{j_{l}}\}_{l=l_{H-2s-T_h+1},\cdots,l_{H-s}}\big)\notag\\
&\ -H\big(\boldsymbol{g}_t|(\boldsymbol{g}_i, \boldsymbol{r}_i:i\in[E]\backslash\{t\}),\{W_i^{j_{l}}\}_{l=1,\cdots,H-s}\big)\nonumber\\
&\overset{(b)}{=}H\big(\boldsymbol{g}_t|(\boldsymbol{g}_i, \boldsymbol{r}_i:i\in[E]\backslash\{t\}), \{W_i^{j_{l}}\}_{l=l_{H-2s-T_h+1},\cdots,l_{H-s}}\big)\notag\\
&\ {=}H\big(\boldsymbol{g}_t|(\boldsymbol{g}_i, \boldsymbol{r}_i:i\in[E]\backslash\{t\})\big)\notag\\
&\ -I\big(\{W_i^{j_{l}}\}_{l=l_{H-2s-T_h+1},\cdots,l_{H-s}}; \boldsymbol{g}_t|(\boldsymbol{g}_i, \boldsymbol{r}_i:i\in[E]\backslash\{t\})\big)\notag\\
&\overset{(c)}{\geq} H\big(\boldsymbol{g}_t|(\boldsymbol{g}_i, \boldsymbol{r}_i:i\in[E]\backslash\{t\})\big)\notag\\
&\ -I\big(\{{\bf{u}}^{j_{l}}\}_{l=l_{H-2s-T_h+1},\cdots,l_{H-s}}; \boldsymbol{g}_t|(\boldsymbol{g}_i, \boldsymbol{r}_i:i\in[E]\backslash\{t\})\big)\notag\\
&\ {=}H\big(\boldsymbol{g}_t|(\boldsymbol{g}_i, \boldsymbol{r}_i:i\in[E]\backslash\{t\})\big)\notag\\
&\ -I\big(\{{\bf{u}}^{j_{l}}\}_{l=l_{H-2s-T_h+1},\cdots,l_{H-2s}}; \boldsymbol{g}_t|(\boldsymbol{g}_i, \boldsymbol{r}_i:i\in[E]\backslash\{t\})\big)\notag\\
&\ -I\big(\{{\bf{u}}^{j_{l}}\}_{l=l_{H-2s+1},\cdots,l_{H-s}}; \boldsymbol{g}_t|(\boldsymbol{g}_i, \boldsymbol{r}_i:i\in[E]\backslash\{t\}),\notag\\
&\qquad \{{\bf{u}}^{j_{l}}\}_{l=l_{H-2s-T_h+1},\cdots,l_{H-2s}}\big)\notag\\
&\overset{(d)}{=}H\big(\boldsymbol{g}_t|(\boldsymbol{g}_i:i\in[E]\backslash\{t\})\big)\notag\\
&\ -I\big(\{{\bf{u}}^{j_{l}}\}_{l=l_{H-2s+1},\cdots,l_{H-s}}; \boldsymbol{g}_t|(\boldsymbol{g}_i, \boldsymbol{r}_i:i\in[E]\backslash\{t\}),\notag\\
&\qquad \{{\bf{u}}^{j_{l}}\}_{l=l_{H-2s-T_h+1},\cdots,l_{H-2s}}\big),
\end{align}
where $(a)$ holds because $W_i^{j_{l_1}},\cdots, W_i^{j_{l_{H-2s-T_h}}}$ are functions of $(\boldsymbol{g}_i, \boldsymbol{r}_i:i\in[E])$; $(b)$ holds because Client $i$ can recover the sum of gradients $\boldsymbol{g}_D$ by $\{W_i^{j_{l}}\}_{l=1,\cdots,H-s}$; $(c)$ holds because $W_i^{j_l}$ is a function of ${\bf{u}}^{j_{l}}$ for $l=l_{H-2s-T_h+1},\cdots,l_{H-s}$; $(d)$ holds by privacy constraints \eqref{pr1}.

On the other hand, we have
\begin{align}\label{endeq3}
H&\big(W_i^{j_{l_1}},\cdots, W_i^{j_{l_{H-2s-T_h}}}|(\boldsymbol{g}_i, \boldsymbol{r}_i:i\in[E]\backslash\{t\})\big)\notag\\
&\leq H\big({\bf{u}}^{j_{l_1}}  ,\ldots,{\bf{u}}^{j_{l_{H-2s-T_h}}}|(\boldsymbol{g}_i,\boldsymbol{r}_i:i\in[E],i\neq t)\big)\notag\\
&\leq H\big(\{\mathcal{F}_t^{j_{l}}(\boldsymbol{g}_t,\boldsymbol{r}_t)\}_{l=l_1,\cdots,l_{H-2s-T_h}}|(\boldsymbol{g}_i,\boldsymbol{r}_i:i\in[E]\backslash\{t\})\big)\notag\\
&\leq H\big(\mathcal{F}_t^{j_{l_1}}(\boldsymbol{g}_t,\boldsymbol{r}_t),\ldots,\mathcal{F}_t^{j_{l_{H-2s-T_h}}}(\boldsymbol{g}_t,\boldsymbol{r}_t)\big).
\end{align} Note that \eqref{endeq3} holds for any $j_{l_1},\cdots,j_{l_{H-2s-T_h}}\in[H]$, we consider a random variable $J\triangleq(J_{l_1},\cdots,J_{l_{H-2s-T_h}})$, where $J_{l_1},\cdots,J_{l_{H-2s-T_h}}\in[H]$, and we assume that $\mathbb{P}(J=J_l)=\frac{1}{{H\choose H-2s-T_h}}$, i.e., $J$ means to select $H-2s-T_h$ numbers from $[H]$ with the same probability. Then, for some $j_{l_1},\cdots,j_{l_{H-2s-T_h}}\in[H]$, we have
\begin{align}\label{endeq4}
H&\big(W_i^{j_{l_1}},\cdots, W_i^{j_{l_{H-2s-T_h}}}|(\boldsymbol{g}_i, \boldsymbol{r}_i:i\in[E]\backslash\{t\})\big)\notag\\
&\leq H\big(\mathcal{F}_t^{j_{l_1}}(\boldsymbol{g}_t,\boldsymbol{r}_t),\ldots,\mathcal{F}_t^{j_{l_{H-2s-T_h}}}(\boldsymbol{g}_t,\boldsymbol{r}_t)\big)\notag\\
&{\leq}\sum_{l=1}^{{H\choose H-2s-T_h}}{\mathbb{P}(J=J_{l})H\big(\mathcal{F}_t^{J_l}(\boldsymbol{g}_t,\boldsymbol{r}_t)\big)}\notag\\
&=H\big(\mathcal{F}_t^{J_{l}}(\boldsymbol{g}_{t},\boldsymbol{r}_t)|J\big)\notag\\
&\overset{(a)}{\leq}\frac{H-2s-T_h}{H}\sum_{k=1}^{H}{H\big(\mathcal{F}_t^{k}(\boldsymbol{g}_t,\boldsymbol{r}_t)\big)},
\end{align}
where $(a)$ holds by Lemma \ref{Lema1}.

Combining \eqref{endeq2} and \eqref{endeq4}, we have
\begin{align}\label{endeq5}
&\frac{H-2s-T_h}{H}\sum_{k=1}^{H}{H\big(\mathcal{F}_t^{k}(\boldsymbol{g}_t,\boldsymbol{r}_t)\big)}\notag\\
&{\geq}H\big(\boldsymbol{g}_t|(\boldsymbol{g}_i:i\in[E]\backslash\{t\})\big)\notag\\
&\ -I\big(\{{\bf{u}}^{j_{l}}\}_{l=l_{H-2s+1},\cdots,l_{H-s}}; \boldsymbol{g}_t|(\boldsymbol{g}_i, \boldsymbol{r}_i:i\in[E]\backslash\{t\}),\notag\\
&\qquad \{{\bf{u}}^{j_{l}}\}_{l=l_{H-2s-T_h+1},\cdots,l_{H-2s}}\big)\notag\\
&{\geq}H\big(\boldsymbol{g}_t|(\boldsymbol{g}_i:i\in[E]\backslash\{t\})\big)\notag\\
&\ -I\big(\{\mathcal{L}_t^{j_{l}}(\boldsymbol{c}_{t,j_{l}})\}_{l=l_{H-2s+1},\cdots,l_{H-s}}; \boldsymbol{g}_t|(\boldsymbol{g}_i, \boldsymbol{r}_i:i\in[E]\backslash\{t\}\big),\notag\\
&\qquad \{{\bf{u}}^{j_{l}}\}_{l=l_{H-2s-T_h+1},\cdots,l_{H-2s}}\big).
\end{align}

By \eqref{comupdef}, we also have

\begin{align}\label{lowbound3}
C_\textnormal{up}&=\max_{i\in[E]} \frac{q_i}{p}\notag\\
&\geq \frac{q_t}{p}\notag\\
&=\frac{q_t\log|\mathbb{F}|}{p\log|\mathbb{F}|}\notag\\
&{\geq}\frac{\sum_{k=1}^{H}{H(\mathcal{F}_t^{k}(\boldsymbol{g}_{t}))}}{p\log|\mathbb{F}|}\notag\\
&\overset{(a)}{=} \frac{\sum_{k=1}^{H}{H(\mathcal{F}_t^{k}(\boldsymbol{g}_{t}))}}{H(\boldsymbol{g}_{t}|(\boldsymbol{g}_i:i\in[E]\backslash\{t\}))}\notag\\
&\overset{(b)}{\geq}\frac{H}{H-2s-T_h}(1\notag\\
&\ -\frac{1}{p\log|\mathbb{F}|}I(\{\mathcal{L}_t^{j_{l}}(\boldsymbol{c}_{t,j_{l}})\}_{l=l_{H-2s+1},\cdots,l_{H-s}}; \boldsymbol{g}_t|\notag\\
&\qquad \{{\bf{u}}^{j_{l}}\}_{l=l_{H-2s-T_h+1},\cdots,l_{H-2s}}),
\end{align}
where $(a)$ holds because $\boldsymbol{g}_t$ is a uniform distribution over the $\mathbb{F}^p$ and $\{\boldsymbol{g}_i\}_{i\in[E]}$ are independent; $(b)$ follows from \eqref{endeq5}.

Therefore, we have 
\begin{align}\label{end2}
C&_\textnormal{up}^{*}\notag\\
&\geq\sup_{m\in\Omega(s)}\frac{H}{H-2s-T_h}(1\notag\\
&\ -\frac{1}{p\log|\mathbb{F}|}I(\{\mathcal{L}_t^{j_{l}}(\boldsymbol{c}_{t,j_{l}})\}_{l=l_{H-2s+1},\cdots,l_{H-s}}; \boldsymbol{g}_t|\notag\\
&\qquad \{{\bf{u}}^{j_{l}}\}_{l=l_{H-2s-T_h+1},\cdots,l_{H-2s}})\notag\\
&=\frac{H}{H-2s-T_h}(1\notag\\
&\ -\inf_{m\in\Omega(s)}\frac{1}{p\log|\mathbb{F}|}I(\{\mathcal{L}_t^{j_{l}}(\boldsymbol{c}_{t,j_{l}})\}_{l=l_{H-2s+1},\cdots,l_{H-s}}; \boldsymbol{g}_t|\notag\\
&\qquad \{{\bf{u}}^{j_{l}}\}_{l=l_{H-2s-T_h+1},\cdots,l_{H-2s}})\notag\\
&\overset{(a)}{=}\frac{H}{H-2s-T_h}.
\end{align}
where $(a)$ holds because these servers do not receive $\boldsymbol{g}_t$ of Client $t$ under some failure scenarios $m$, i.e., $\exists m\in\Omega(s),$
$$I(\{\mathcal{L}_t^{j_{l}}(\boldsymbol{c}_{t,j_{l}})\}_{l=l_{H-2s+1},\cdots,l_{H-s}}; \boldsymbol{g}_t|\\ \{{\bf{u}}^{j_{l}}\}_{l=l_{H-2s-T_h+1},\cdots,l_{H-2s}})$$ is zero.

By \eqref{rCl}, we have
\begin{align}\label{lowbound4}
C&_\textnormal{down}\notag\\
&=\max_{l\in[E]} \frac{d_l}{p}\notag\\
&=\frac{d_i\log|\mathbb{F}|}{p\log|\mathbb{F}|}\notag\\
&=\frac{\sum_{l=1}^{H-s} d^{j_l}_i\log|\mathbb{F}|}{p\log|\mathbb{F}|}\notag\\
&{\geq}\frac{\sum_{l=1}^{H-s}H(W^{j_l}_i)}{p\log|\mathbb{F}|}\notag\\
&\overset{(a)}{\geq}\frac{\sum_{l=1}^{H-2s}H(W^{j_l}_i)}{p\log|\mathbb{F}|}\notag\\
&\overset{(b)}{\geq}\frac{H-2s}{H-2s-T_h}\frac{H(W^{j_{l_1}}_i,\ldots,W^{j_{l_{H-2s-T_h}}}_i)}{p\log|\mathbb{F}|}\notag\\
&\overset{(c)}{\geq}\frac{H-2s}{H-2s-T_h}(1\notag\\
&\ -\frac{I\big(\{\mathcal{L}_t^{j_{l}}(\boldsymbol{c}_{t,j_{l}})\}_{l=l_{H-2s+1},\cdots,l_{H-s}};\boldsymbol{g}_t|\{{\bf{u}}^{j_{l}}\}_{l=l_{H-2s-T_h+1},\cdots,l_{H-2s}}\big)}{p\log|\mathbb{F}|}),
\end{align}
where $(a)$ holds by the non-negativity of entropy; $(b)$ uses Lemma \ref{Lema1} again; $(c)$ follows from \eqref{endeq2}.

Therefore, we have 
\begin{align}\label{end3}
C&_\textnormal{down}^{*}\notag\\
&\geq\sup_{m\in\Omega(s)}\frac{H-2s}{H-2s-T_h}(1\notag\\
&\ -\frac{I\big(\{\mathcal{L}_t^{j_{l}}(\boldsymbol{c}_{t,j_{l}})\}_{l=l_{H-2s+1},\cdots,l_{H-s}};\boldsymbol{g}_t|\{{\bf{u}}^{j_{l}}\}_{l=l_{H-2s-T_h+1},\cdots,l_{H-2s}}\big)}{p\log|\mathbb{F}|})\notag\\
&=\frac{H-2s}{H-2s-T_h}(1-\notag\\
&\ -\inf_{m\in\Omega(s)}\frac{1}{p\log|\mathbb{F}|}I\big(\{\mathcal{L}_t^{j_{l}}(\boldsymbol{c}_{t,j_{l}})\}_{l=l_{H-2s+1},\cdots,l_{H-s}};\boldsymbol{g}_t|\notag\\
&\qquad \{{\bf{u}}^{j_{l}}\}_{l=l_{H-2s-T_h+1},\cdots,l_{H-2s}}\big)\notag\\
&\ {=}\frac{H-2s}{H-2s-T_h}.
\end{align}
\subsection{Proof of Security}\label{secproof}

For a federated learning setting with $E$ clients, $H$ servers, {length  of gradient $p$} and a resiliency threshold of $s<\frac{H}{2}$, the secure aggregation protocol presented in Section \ref{LCM} is $(T_h,T_c) $-private.

For any $T_c$ collusion clients, in addition to their own stored gradient values $\{\boldsymbol{g}_i\}_{i\in\mathcal{T}_c}$ and random noise $\{\boldsymbol{s}_{ij}\}_{i\in\mathcal{T}_c \textnormal{or} j\in\mathcal{T}_c }$, they can also obtain messages which are functions of $\{\boldsymbol{y}_{l_1},\cdots,\boldsymbol{y}_{l_{E-T_c}}\}$ where $l_{i}\in[E]\backslash\mathcal{T}_c$ for $i\in[E-T_c]$ and $l_1<\cdots<l_{E-T_c}$. Therefore, we have 
\begin{align}
I&(\{\boldsymbol{g}_i\}_{i\in[E]\backslash\mathcal{T}_c}; \mathtt{VIEW}_{\mathcal{T}_c}|\boldsymbol{g}_D,\{\boldsymbol{g}_l\}_{l\in{\mathcal{T}_c}})\notag\\
&\leq I(\{\boldsymbol{g}_i\}_{i\in[E]\backslash\mathcal{T}_c}; \{\boldsymbol{y}_i\}_{i\in[E]},\{\boldsymbol{s}_{ij}\}_{i\in\mathcal{T}_c \textnormal{or} j\in\mathcal{T}_c }|\{\boldsymbol{g}_l\}_{l\in{\mathcal{T}_c}},\boldsymbol{g}_D)\notag\\
&\ {=}I(\{\boldsymbol{g}_i\}_{i\in[E]\backslash\mathcal{T}_c};\{\boldsymbol{s}_{ij}\}_{i\in\mathcal{T}_c \textnormal{or} j\in\mathcal{T}_c }|\{\boldsymbol{g}_l\}_{l\in{\mathcal{T}_c}},\boldsymbol{g}_D)\notag\\
&\ + I(\{\boldsymbol{g}_i\}_{i\in[E]\backslash\mathcal{T}_c}; \{\boldsymbol{y}_i\}_{i\in[E]}|\{\boldsymbol{g}_l\}_{l\in{\mathcal{T}_c}},\boldsymbol{g}_D,\{\boldsymbol{s}_{ij}\}_{i\in\mathcal{T}_c \textnormal{or} j\in\mathcal{T}_c })\notag\\
&\overset{(a)}{=}I(\{\boldsymbol{g}_i\}_{i\in[E]\backslash\mathcal{T}_c}; \{\boldsymbol{y}_i\}_{i\in[E]}|\{\boldsymbol{g}_l\}_{l\in{\mathcal{T}_c}},\boldsymbol{g}_D,\{\boldsymbol{s}_{ij}\}_{i\in\mathcal{T}_c \textnormal{or} j\in\mathcal{T}_c })\notag\\
&\ {=}I(\{\boldsymbol{g}_i\}_{i\in[E]\backslash\mathcal{T}_c\cup\{l_1\}}; \{\boldsymbol{y}_i\}_{i\in[E]}|\notag\\
&\qquad\{\boldsymbol{g}_l\}_{l\in{\mathcal{T}_c}},\boldsymbol{g}_D,\{\boldsymbol{s}_{ij}\}_{i\in\mathcal{T}_c \textnormal{or} j\in\mathcal{T}_c })\notag\\
&\ +I(\boldsymbol{g}_{l_1}; \{\boldsymbol{y}_i\}_{i\in[E]}|\{\boldsymbol{g}_i\}_{i\in[E]\backslash\{l_1\}},\boldsymbol{g}_D,\{\boldsymbol{s}_{ij}\}_{i\in\mathcal{T}_c \textnormal{or} j\in\mathcal{T}_c })\notag\\
&\overset{(b)}{=}I(\{\boldsymbol{g}_i\}_{i\in[E]\backslash\mathcal{T}_c\cup\{l_1\}}; \{\boldsymbol{y}_i\}_{i\in[E]}|\notag\\
&\qquad\{\boldsymbol{g}_l\}_{l\in{\mathcal{T}_c}},\boldsymbol{g}_D,\{\boldsymbol{s}_{ij}\}_{i\in\mathcal{T}_c \textnormal{or} j\in\mathcal{T}_c })\notag\\
&\ {=}I(\{\boldsymbol{g}_i\}_{i\in[E]\backslash\mathcal{T}_c\cup\{l_1\}}; \{\boldsymbol{y}_i\}_{i\in[E]\backslash\{l_1\}}|\notag\\
&\qquad\{\boldsymbol{g}_l\}_{l\in{\mathcal{T}_c}},\boldsymbol{g}_D,\{\boldsymbol{s}_{ij}\}_{i\in\mathcal{T}_c \textnormal{or} j\in\mathcal{T}_c })\notag\\
&\ +I(\{\boldsymbol{g}_i\}_{i\in[E]\backslash\mathcal{T}_c\cup\{l_1\}}; \boldsymbol{y}_{l_1}|\notag\\
&\qquad\{\boldsymbol{g}_l\}_{l\in{\mathcal{T}_c}},\boldsymbol{g}_D,\{\boldsymbol{s}_{ij}\}_{i\in\mathcal{T}_c \textnormal{or} j\in\mathcal{T}_c },\{\boldsymbol{y}_i\}_{i\in[E]\backslash\{l_1\}})\notag\\
&\overset{(c)}{=} I(\{\boldsymbol{g}_i\}_{i\in[E]\backslash\mathcal{T}_c\cup\{l_1\}}; \{\boldsymbol{y}_i\}_{i\in[E]\backslash\{l_1\}}|\notag\\
&\qquad\{\boldsymbol{g}_l\}_{l\in{\mathcal{T}_c}},\boldsymbol{g}_D,\{\boldsymbol{s}_{ij}\}_{i\in\mathcal{T}_c \textnormal{or} j\in\mathcal{T}_c })\notag\\
&\overset{(d)}{=}I(\{\boldsymbol{g}_i\}_{i\in[E]\backslash\mathcal{T}_c\cup\{l_1\}}; \{\boldsymbol{y}_i\}_{i\in[E]\backslash\mathcal{T}_c\cup\{l_1\}}|\notag\\
&\qquad\{\boldsymbol{g}_l\}_{l\in{\mathcal{T}_c}},\boldsymbol{g}_D,\{\boldsymbol{s}_{ij}\}_{i\in\mathcal{T}_c \textnormal{or} j\in\mathcal{T}_c })\notag\\
&\ {=} H(\{\boldsymbol{y}_i\}_{i\in[E]\backslash\mathcal{T}_c\cup\{l_1\}}|\{\boldsymbol{g}_l\}_{l\in{\mathcal{T}_c}},\boldsymbol{g}_D,\{\boldsymbol{s}_{ij}\}_{i\in\mathcal{T}_c \textnormal{or} j\in\mathcal{T}_c })\notag\\
&\ -H(\{\boldsymbol{y}_i\}_{i\in[E]\backslash\mathcal{T}_c\cup\{l_1\}}|\{\boldsymbol{g}_i\}_{i\in[E]\backslash\{l_1\}},\boldsymbol{g}_D,\{\boldsymbol{s}_{ij}\}_{i\in\mathcal{T}_c \textnormal{or} j\in\mathcal{T}_c })\notag\\
&\ {\leq} H(\{\boldsymbol{y}_i\}_{i\in[E]\backslash\mathcal{T}_c\cup\{l_1\}}|\{\boldsymbol{g}_l\}_{l\in{\mathcal{T}_c}},\boldsymbol{g}_D,\{\boldsymbol{s}_{ij}\}_{i\in\mathcal{T}_c \textnormal{or} j\in\mathcal{T}_c })\notag\\
&\ -H(\{\boldsymbol{y}_i\}_{i\in[E]\backslash\mathcal{T}_c\cup\{l_1\}}|\{\boldsymbol{g}_i\}_{i\in[E]\backslash\{l_1\}},\boldsymbol{g}_D,\{\boldsymbol{s}_{ij}\}_{i\neq l_1 })\notag\\
&\ {=} H(\{\boldsymbol{y}_i\}_{i\in[E]\backslash\mathcal{T}_c\cup\{l_1\}}|\{\boldsymbol{g}_l\}_{l\in{\mathcal{T}_c}},\boldsymbol{g}_D,\{\boldsymbol{s}_{ij}\}_{i\in\mathcal{T}_c \textnormal{or} j\in\mathcal{T}_c })\notag\\
&\ -H(\boldsymbol{s}_{l_1l_2},\cdots,\boldsymbol{s}_{l_1l_{E-T_c}}|\{\boldsymbol{g}_i\}_{i\in[E]\backslash\{l_1\}},\boldsymbol{g}_D,\{\boldsymbol{s}_{ij}\}_{i\neq l_1 })\notag\\
&\overset{(e)}{=} H(\{\boldsymbol{y}_i\}_{i\in[E]\backslash\mathcal{T}_c\cup\{l_1\}}|\{\boldsymbol{g}_l\}_{l\in{\mathcal{T}_c}},\boldsymbol{g}_D,\{\boldsymbol{s}_{ij}\}_{i\in\mathcal{T}_c \textnormal{or} j\in\mathcal{T}_c })\notag\\
&\ -H(\boldsymbol{s}_{l_1l_2},\cdots,\boldsymbol{s}_{l_1l_{E-T_c}})\notag\\
&\ {\leq}\sum_{i=2}^{E-T_c} (H(\boldsymbol{y}_{l_i})-H(\boldsymbol{s}_{l_1l_i}))\leq 0,
\end{align}
where $(a)$ holds because $\{\boldsymbol{s}_{ij}\}_{i\in\mathcal{T}_c \textnormal{or} j\in\mathcal{T}_c }$ and $\{\boldsymbol{g}_{i}\}_{i\in[E]}$ are independent; $(b)$ and $(c)$ holds by $\boldsymbol{g}_D=\sum_{i=1}^{E}\boldsymbol{g}_i$ and $\boldsymbol{g}_D=\sum_{i=1}^{E}\boldsymbol{y}_i$, respectively; $(c)$ and $(d)$  follow from \eqref{lccm}; $(e)$ holds by the independence of $\boldsymbol{s}_{l_1l_2},\cdots,\boldsymbol{s}_{l_1l_{E-T_c}}$; the last inequality holds because $\boldsymbol{y}_{l_i}\in\mathbb{F}^p$ and $\boldsymbol{s}_{l_1l_j}$ is uniformly distributed over $\mathbb{F}^p$ for $i\in[E-T_c]$.

For any $T_h$ collusion servers denoted by $\{l_1,\cdots,l_{T_h}\}$, we have
\begin{align}
&\begin{pmatrix}
u_1(\alpha_{l_1})&\cdots&u_E(\alpha_{l_1})\\
u_1(\alpha_{l_2})&\cdots& u_E(\alpha_{l_1})\\
\vdots&\cdots&\vdots\\
u_1(\alpha_{l_{T_h}})&\cdots&u_E(\alpha_{l_1})\\
\end{pmatrix}\\
&=\begin{pmatrix}\label{collm}
U_{{l_1}1}&\cdots&U_{{l_1}(k+T_h)}\\
U_{{l_2}1}&\cdots&U_{{l_2}(k+T_h)} \\
\vdots&\ddots&\vdots\\
U_{{l_{T_h}}1}&\cdots&U_{{l_{T_h}}(k+T_h)}\\
\end{pmatrix}
\begin{pmatrix}
\boldsymbol{y}_{1,1}&\cdots&\boldsymbol{y}_{E,1}\\
\boldsymbol{y}_{1,2}&\cdots&\boldsymbol{y}_{E,2}\\
\vdots&\cdots&\vdots\\
\boldsymbol{y}_{1,k}&\cdots&\boldsymbol{y}_{E,k}\\
\boldsymbol{Z}_{1,1}&\cdots&\boldsymbol{Z}_{E,1}\\
\boldsymbol{Z}_{1,2}&\cdots&\boldsymbol{Z}_{E,1}\\
\vdots&\cdots&\vdots\\
\boldsymbol{Z}_{1,T_h}&\cdots&\boldsymbol{Z}_{E,T_h}\\
\end{pmatrix}\\
&=\begin{pmatrix}
U_A&U_B
\end{pmatrix}
\begin{pmatrix}
\boldsymbol{Y}\\
\boldsymbol{Z}
\end{pmatrix}=U_A\boldsymbol{Y}+U_B\boldsymbol{Z},
\end{align}
where $U_A$ denotes the first $k$ columns of the left matrix in (\ref{collm}), $U_B$ denotes the last $T_h$ columns in (\ref{collm}), $\boldsymbol{Y}$ and $\boldsymbol{Z}$ denote the first $k$ rows and the last $T_h$ rows of the right matrix in (\ref{collm}), respectively. Therefore, we have

\begin{align}
I&(\{\boldsymbol{g}_i\}_{i\in[E]}; \mathtt{VIEW}_{\mathcal{T}_h})\notag\\
&\overset{(a)}{\leq} I(\{\boldsymbol{g}_i\}_{i\in[E]}; \{(u_1(\alpha_l),\cdots,u_E(\alpha_l))\}_{l\in{\mathcal{T}_h}})\notag\\
&\overset{(b)}{=}I(\{\boldsymbol{g}_i\}_{i\in[E]}; U_A\boldsymbol{Y}+U_B\boldsymbol{Z})\notag\\
&\overset{(c)}{=}I(\{\boldsymbol{g}_i\}_{i\in[E]}; U_B^{-1}U_A\boldsymbol{Y}+\boldsymbol{Z})\notag\\
&\ {=}H(U_B^{-1}U_A\boldsymbol{Y}+\boldsymbol{Z})-H(U_B^{-1}U_A\boldsymbol{Y}+\boldsymbol{Z}|\{\boldsymbol{g}_i\}_{i\in[E]})\notag\\
&\ {\leq} H(U_B^{-1}U_A\boldsymbol{Y}+\boldsymbol{Z})\notag\\
&\ -H(U_B^{-1}U_A\boldsymbol{Y}+\boldsymbol{Z}|\{\boldsymbol{g}_i\}_{i\in[E]}, \{\boldsymbol{s}_i\}_{i\in[E]})\notag\\
&\overset{(d)}{=}H(U_B^{-1}U_A\boldsymbol{Y}+\boldsymbol{Z})-H(\boldsymbol{Z})\overset{(e)}{\leq} 0,
\end{align}
where $(a)$ holds because all messages that $T_h$ colluding servers can obtain are functions of $\{(u_1(\alpha_l),\cdots,u_E(\alpha_l))\}_{l\in{\mathcal{T}_h}}$; $(b)$ follows from \eqref{collm}; in $(c)$, we use the fact that $U_B$ is an invertible matrix due to the characteristics of lagrangian interpolation; $(d)$ holds because $\boldsymbol{y}_{i,r}=\boldsymbol{g}_{i,r}+\boldsymbol{s}_{i,r}$ for $i\in[E],j\in[k]$ and $\boldsymbol{Z}$ and $\{\boldsymbol{g}_i\}_{i\in[E]}, \{\boldsymbol{s}_i\}_{i\in[E]}$ are independent; $(e)$ holds because $U_B^{-1}U_A\boldsymbol{Y}+\boldsymbol{Z}$ and $\boldsymbol{Z}$ are in $\mathbb{F}^{\frac{pT_h}{k}\times E}$ and $\boldsymbol{Z}$ is a uniform distribution over $\mathbb{F}^{\frac{pT_h}{k}\times E}$.
\subsection{Proof of Lemma 1}\label{l1proof}
We first prove $T_h\leq H-2s-1$ in (\ref{cup2}), when a secure aggregation protocol is $(T_h,T_c) $-private and $s$-resiliency.  Recall our threat model for servers: any set of at most $T_h$ colluding servers cannot get any information about the client’s inputs. Suppose $T_h\geq H-2s$, we have 
\begin{equation}\label{sec12}
I(\boldsymbol{g}_1+\boldsymbol{g}_2+\cdots+\boldsymbol{g}_E;\mathtt{VIEW}_{\mathcal{T}})=0,
\end{equation}
where $\mathcal{T}$ denotes any set of $T\leq T_h$ servers. We consider a special failure table that for a client, successfully links $H-s$ servers, but the data sent by other clients are only received by $H-2s$ servers in these nodes. In other words, other clients' data are not successfully sent to the same $s$ servers. Then, the client needs to recover the aggregated value from all the messages these $H-2s$ server nodes receive. However, by \eqref{pr1}, to satisfy privacy against any $H-2s$ servers
\begin{equation}\label{sec22}
I(\boldsymbol{g}_1+\boldsymbol{g}_2+\cdots+\boldsymbol{g}_E;\mathtt{VIEW}_{\mathcal{T}'})=0,
\end{equation}
where $\mathcal{T}'$ denotes these $H-2s$ servers. Therefore, this client cannot recover the final aggregated result, which violates the definition of $s$-resiliency.

Next, we prove $T_c\leq E-2$. Suppose $T_c = E-1$, if any $E-1$ clients collude, the data of the remaining client must be obtained according to the aggregated value. Therefore, we have $T_c < E-1$, i.e., $T_c \leq E-2$.


\end{document}